\documentclass{article}

\usepackage{arxiv}

\usepackage[utf8]{inputenc} 
\usepackage[T1]{fontenc}    
\usepackage{hyperref}       
\usepackage{url}            
\usepackage{booktabs}       
\usepackage{amsfonts}       
\usepackage{nicefrac}       
\usepackage{microtype}      
\usepackage{lipsum}		
\usepackage{graphicx}
\usepackage[authoryear,sort]{natbib}
\usepackage{doi}

\usepackage{amssymb}
\usepackage{amsthm}

\usepackage{amsmath}
\usepackage{url}
\usepackage{xcolor}
\usepackage{multirow}
\usepackage{float}
\usepackage{scalerel}
\usepackage{bm}
\usepackage{enumerate}

\title{An F5 Algorithm\\for Tropical Gr\"{o}bner Bases in the Weyl Algebras}

\author{Ari Dwi Hartanto\thanks{Corresponding author} \\
	Department of Mathematics,\\
	Faculty of Mathematics and Natural Sciences,\\ Universitas Gadjah Mada, INDONESIA\\
	\texttt{ari@ugm.ac.id} \\
	\And
	Katsuyoshi Ohara \\
	Faculty of Mathematics and Physics,\\
	Kanazawa University, JAPAN \\
	\texttt{ohara@se.kanazawa-u.ac.jp} \\
}

\renewcommand{\headeright}{{\footnotesize A.D. Hartanto \& K. Ohara}}
\renewcommand{\shorttitle}{\textit{An F5 Algorithm for Tropical GB in the Weyl Algebras}}

\hypersetup{
pdftitle={On F5 Algorithm for Tropical Groebner Bases on Weyl Algebras},
pdfsubject={math},
pdfauthor={Ari Dwi~Hartanto, Katsuyoshi~Ohara},
pdfkeywords={Gr\"{o}bner basis, signature, syzygy, tropical term order, valuation, Weyl algebra},
}

\newtheoremstyle{dotlesstheorem}  
  {1.75ex}   
  {0.5ex}   
  {\itshape}  
  {0pt}       
  {\bfseries} 
  {}         
  {  }  
  {}          
\theoremstyle{dotlesstheorem}
\newtheorem{thm}{Theorem}
\newtheorem{lem}[thm]{Lemma}
\newtheorem{prop}[thm]{Proposition}
\newtheorem{cor}[thm]{Corollary}

\newtheoremstyle{dotlessexample}  
  {1.75ex}   
  {0.5ex}   
  {}  
  {0pt}       
  {\bfseries} 
  {}         
  {  }  
  {}          
\theoremstyle{dotlessexample}
\newtheorem{defn}{Definition}
\newtheorem{eg}{Example}
\newtheorem{rmk}{Remark}

\newcommand{\Z}{\mathbb{Z}}
\newcommand{\Q}{\mathbb{Q}}
\newcommand{\R}{\mathbb{R}}

\newcommand{\val}{{\rm val}}
\newcommand{\x}{{\rm x}}
\newcommand{\supp}{{\rm supp}}
\newcommand{\LT}{\mbox{{\rm LT}\:\!}}
\newcommand{\LM}{\mbox{{\rm LM}\:\!}}
\newcommand{\LC}{\mbox{{\rm LC}\:\!}}
\newcommand{\LTh}{{\rm LT}_h}
\newcommand{\LMh}{{\rm LM}_h}
\newcommand{\LCh}{{\rm LC}_h}

\renewcommand{\d}{\partial}
\newcommand{\dx}{\partial_x}
\newcommand{\dy}{\partial_y}
\newcommand{\dz}{\partial_z}
\newcommand{\mon}{{\rm mon}}

\newcommand{\End}{{\rm End}}
\newcommand{\syz}[1]{{\rm Syz}(#1)}
\newcommand{\lsyz}[1]{{\rm LSyz}(#1)}

\newcommand{\ve}{\mathbf{e}}
\renewcommand{\a}{\bm{a}}
\newcommand{\f}{\bm{f}}
\newcommand{\g}{\bm{g}}
\newcommand{\h}{\bm{h}}

\newcommand{\q}{\bm{q}}
\renewcommand{\r}{\bm{r}}

\newcommand{\sn}{\scaleto{N}{5pt}}
\newcommand{\mn}{\scaleto{N}{7pt}}
\newcommand{\Psih}{\Psi_{\,\!\!\scaleto{h}{4pt}}}
\newcommand{\s}{\mathfrak{s}}

\usepackage[ruled,linesnumbered]{algorithm2e}
\SetKwInput{KwInput}{Input}
\SetKwInput{KwOutput}{Output}
\makeatletter
\newcommand{\nosemic}{\renewcommand{\@endalgocfline}{\relax}}
\newcommand{\dosemic}{\renewcommand{\@endalgocfline}{\algocf@endline}}
\let\oldnl\nl
\newcommand{\nonl}{\renewcommand{\nl}{\let\nl\oldnl}}
\makeatother

\makeatletter
\renewenvironment{proof}[1][\proofname]{\par
  \vspace{-0.5ex}
  \pushQED{\qed}%
  \normalfont
  \topsep0pt \partopsep0pt 
  \trivlist
  \item[\hskip\labelsep
        \itshape
    #1\@addpunct{.}]\ignorespaces
}{%
  \popQED\endtrivlist\@endpefalse
  \addvspace{0.5ex plus 0.35ex} 
}
\makeatother

\begin{document}
\maketitle

\begin{abstract}
A Gr\"{o}bner basis computation for the Weyl algebra with respect to a tropical term order and by using a homogenization-dehomogenization technique is sufficiently sluggish. A significant number of reductions to zero occur. To improve the computation, a tropical F5 algorithm is developed for this context. As a member of the family of signature-based algorithms, this algorithm keeps track of where Weyl algebra elements come from to anticipate reductions to zero. The total order for ordering module monomials or signatures in this paper is designed as close as possible to the definition of the tropical term order. As in \cite{Vaccon2021}, this total order is not compatible with the tropical term order.
\end{abstract}

\vspace{-3ex}
\indent{\keywords{Gr\"{o}bner basis \and signature \and syzygy \and tropical term order \and valuation \and Weyl algebra}

\bigskip

\section{Introduction}

The theory of tropical geometry has developed in the past few decades. Some researchers developed this theory in some aspects; for instance, the aspect of theory and its applications (\cite{DMBS2015}), a theory of tropical differential algebra (\cite{Falkensteiner2020}), and a theory of tropicalization on {D}-ideals (\cite{Hsiao2023}). Tropical geometry can be related to the theory of Gr\"{o}bner basis, especially the computational characterization of tropical variety. A computational theory of tropical Gr\"{o}bner basis was developed by Chan and Maclagan (\cite{ACDM2019}) in 2019. This theory undergoes further development. In \cite{Vaccon2015} and in \cite{Vaccon2021}, Vaccon et.al. presented a tropical F5 algorithm for polynomial rings $K[x_1,\cdots,x_n]$. Furthermore, in \cite{ADH2023}, the computational theory of tropical Gr\"{o}bner basis was extended to the Weyl algebras. A homogenization-dehomogenization technique was used, and the Buchberger algorithm was applied in this computation. We remark that this computational method needs to be enhanced, as too many reductions to zero affect a slow computation. Many of those zero reductions are known as ``useless reductions.''

In 2002, Faug\`{e}re introduced an algorithm called the F5 algorithm (see \cite{Faugere2002}). This algorithm is based on signatures in the form of the leading monomials of the polynomial module representations, and then it is called a signature-based algorithm. It was proven that this algorithm is able to reduce useless reductions. Next, many researchers worked to understand the criterion behind the F5 algorithm and to develop signature-based algorithms. Some results can be seen in: \cite{Eder2008a}, \cite{Eder2008b}, \cite{Eder2010}, \cite{Arri2011}, G2V (\cite{Gao2010a}), GVW (\cite{Gao2010b}), SB (\cite{Roune2012}), and the F5 algorithm in tropical settings (\cite{Vaccon2015},\cite{Vaccon2021}).

In this paper, we present a tropical F5 algorithm to compute a tropical Gr\"{o}bner basis on the homogenized Weyl algebra $D_n^{(h)}\!(K)$ over a field $K$ with a valuation. Based on \cite{ADH2023}, once a Gr\"{o}bner basis $G^{(h)}$ in $D_n^{(h)}\!(K)$ is obtained, we can perform the dehomogenization process to obtain the corresponding Gr\"{o}bner basis $G$ in the Weyl algebra $D_n(K)$. As a signature-based algorithm, the tropical F5 algorithm requires a total order on the module $(D_n^{(h)})^{\sn}$ for ordering module monomials or signatures. We define a total order $\leq_{\rm sign}$ on module $R^{\sn}$ which is almost similar to the definition $\leq_{\rm sign}$ in the paper of \cite{Vaccon2021}. In \cite{Vaccon2021}, the theory utilizes a vector space filtration. The vactor space filtration is mainly used for defining signature of a polynomial. The theory utilizing the vector space filtration requires the commutative property of the used ring, whereas our ring, $D_n^{(h)}\!(K)$, is non-commutative. However, if we deeply observe the flow of the main theory of their F5 algorithm, especially in the F5 criterion, it actually follows the theory of the F5 algorithm in \cite{Arri2011} and the tropical F5 algorithm in \cite{Vaccon2015}. In this paper, we develop an F5 algorithm to compute a tropical Gr\"{o}bner basis in $D_n^{(h)}(K)$ by using a total order $\leq_{\rm sign}$ defined as similar as possible to the definition $\leq_{\rm sign}$ in \cite{Vaccon2021} and by following the flow of the theory in \cite{Arri2011} and in \cite{Vaccon2015}.

\section{Term Ordering and Gr\"{o}bner bases}
\label{sec2}
For convenience of the readers, we begin with a glance introduction to the Weyl algebras.
\vspace{-1ex}
\subsection{The Weyl algebras over fields with valuations}
\vspace{-1ex}
Let $K$ be a field of characteristic $0$. We equip the field $K$ with a valuation which is defined as a function $\val:K^*\!\longrightarrow\! \R$ satisfying $\val(ab)=\val(a)+\val(b)$ and $\val(a+b)\geq\min(\val(a),\val(b))$. For example, in the field of rational functions $\Q(t)$, consider $f/g$ with the Taylor series $at^m+(\text{higher order terms})$. We can define $\val(f/g):=m$, a valuation of $f/g$. Another example is the {\it $p$-adic valuation} defined in the field $\Q$, where $p$ is a prime number. For a non-zero rational number $q$ expressed by $q=p^{m}\frac{a}{b}$, where $m\in\Z_{\geq 0}$ and $a,b\in\Z$, $b\neq0$ such that $p$ does not divide both $a$ and $b$, the $p$-adic valuation of $q$ is $\val(q):=m$.

The Weyl algebra over $K$, denoted by $D_n(K)$ or by $D_n$, is the algebra
$K\langle x_1,\cdots,x_n,\d_1,\cdots,\d_n\rangle$
which is generated by operators $x_i,\d_i\in\End_K(K[\x])$ given by, respectively, $f(\x)\mapsto x_if(\x)$ and $f(\x)\mapsto \frac{\d}{\d x_i}f(\x):=\d_i f(\x)$, $1\leq i\leq n$. These operators satisfy the following relations:
\vspace{-1ex}
\begin{align*}
x_jx_i-x_ix_j=0 \; \text{ and } \; &\d_j\d_i-\d_i\d_j=0\ \ \text{ for } 1\leq i,j\leq n;\\
\d_ix_j-x_j\d_i &= 0\ \qquad\quad\quad \;\;\;\; \text{ for } 1\leq i\neq j\leq n;\\
\d_i x_i-x_i\d_i&=1\ \qquad\quad\quad \;\;\;\; \text{ for } 1\leq i\leq n.
\end{align*}

\vspace{-1ex}
\noindent
By the canonical $K$-vector space basis $\{\x^{\alpha}\d^{\beta}\mid \alpha,\beta\in\Z_{\geq 0}^n\}\subset D_n(K)$, any operator $f\in D_n(K)$ can be written as a finite sum $f=\sum_{\alpha,\beta\in\Z_{\geq 0}^n}\!\!\!c_{\alpha\beta}\x^{\alpha}\d^{\beta}$.

\subsection{Tropical term orders and tropical Gr\"{o}bner bases}
We define a $K$-vector space isomorphism $\Psi:D_n(K) \longrightarrow K[\x,\xi]$ given by $\x^{\alpha}\d^{\beta} \mapsto \x^{\alpha}\xi^{\beta}$.
The notation $K[\x,\xi]$ means the polynomial ring $K[x_1,\cdots,x_n,\xi_1,\cdots,\xi_n]$ with $2n$ indeterminates. For any $f\in D_n(K)$, we call $\Psi(f)$ the \emph{total symbol} of $f$. We denote $\mathcal{M}_{\xi}$ and $\mathcal{T}_{\xi}$, respectively, the set of all monomials $\{\x^{\alpha}\xi^{\beta}\mid \alpha,\beta\in\Z_{\geq 0}^n\}$ and the set of all terms $\{c\x^{\alpha}\xi^{\beta}\mid c\in K^*,\alpha,\beta\in\Z_{\geq 0}^n\}$ in $K[\x,\xi]$.

We recall the definition of \emph{tropical term order} $\leq$ on $\mathcal{T}_{\xi}$  from the paper (in preprint) of \cite{ADH2023} as follows.
\begin{defn}\label{deftroptermorderDn}
Let $\preceq$ be a monomial order on $\mathcal{M}_{\xi}$ and $w,\omega\in \R^{2n}$ such that $w_i\geq 0$ and $(\max_{\scaleto{1\leq j\leq n}{5pt}}\, w_j)<w_{n+i}$, $1\!\leq\! i\!\leq\! n$.
For $a\x^{\alpha}\xi^{\beta},b\x^{\alpha'}\xi^{\beta'}\!\!\in\mathcal{T}_{\xi}$, we write:
\vspace{-0.75ex}
\begin{enumerate}[(1)]
\item $a\x^{\alpha}\xi^{\beta}<b\x^{\alpha'}\xi^{\beta'}$  \ if and only if:
{\small
\begin{enumerate}[(i)]
\vspace{-0.5ex}
\item $w\cdot(\alpha,\beta)<w\cdot(\alpha',\beta')$; or
\item $w\cdot(\alpha,\beta)=w\cdot(\alpha',\beta')$ and $-\val(a)+\omega\cdot(\alpha,\beta)<-\val(b)+\omega\cdot(\alpha',\beta')$; or
\item $w\cdot(\alpha,\beta)=w\cdot(\alpha',\beta')$ and $-\val(a)+\omega\cdot(\alpha,\beta)=-\val(b)+\omega\cdot(\alpha',\beta')$ \ and \ $\x^{\alpha}\xi^{\beta}\prec\x^{\alpha'}\xi^{\beta'}$;
\end{enumerate}
}
\vspace{-0.25ex}
\item $a\x^{\alpha}\xi^{\beta}= b\x^{\alpha'}\xi^{\beta'}$  if and only if $\val(a)=\val(b)$ and $(\alpha,\beta)=(\alpha',\beta')$.
\end{enumerate}
\end{defn}

\smallskip
\noindent
If we restrict our ring to $K[\x]\subset K[\x,\xi]$ and set $w=(\bm{0};\bm{1})$ and $\omega=-1\cdot(u,v)$ for some $u,v\in \R^n$, then $\leq$ is a tropical term order on $K[\x]$ with respect to $u$ as in \cite{ACDM2019}.

Let $f\in D_n$ and $a\x^{\tau}\xi^{\sigma}\in\mathcal{T}_{\xi}$ the biggest term of the total symbol $\Psi(f)$ with respect to $<$. The leading term, the leading monomial, and the leading coefficient of $\Psi(f)$, respectively, is defined as $\LT(f):=a\x^{\tau}\xi^{\sigma}$, $\LM(f):=\x^{\tau}\xi^{\sigma}$, and $\LC(f):=a$. We then define $\LT(f):=\LT(\Psi(f))$, $\LM(f):=\LM(\Psi(f))$, and $\LC(f):=\LC(\Psi(f))$. For a non-empty subset $G\subseteq D_n(K)$, we define $\LM(G):=\{\LM(g)\mid g\in G\}$.

\begin{defn}
A non-empty subset $G\subset D_n(K)$ is called a (\emph{tropical}) \emph{Gr\"{o}bner basis} of a left ideal $I\subset D_n(K)$ if
$
D_n G= I \ \text{and} \ \langle \LM(G)\rangle=\langle\LM(I)\rangle.
$
\end{defn}

We introduce a new variable $h$ which commutes with all $x_i$ and $\d_i$, and define the homogenized Weyl algebra $D_n^{(h)}\!(K):=K[h]\langle x_1,\cdots\!,x_n,\d_1,\cdots\!,\d_n\rangle$ as a ring extension of $D_n(K)$. The homogenization of $f=\!\!\sum_{(\alpha,\beta)\in \mathcal{E}}\!c_{\alpha\beta}\x^{\alpha}\partial^{\beta}\in D_n(K)$, $\mathcal{E}\subset\Z_{\geq 0}^n\times\Z_{\geq 0}^n$ is
\vspace{-1ex}
$$H(f):=\sum_{(\alpha,\beta)\in \mathcal{E}} c_{\alpha\beta}\x^{\alpha}\partial^{\beta}h^{M-|\alpha|-|\beta|}\in D_n^{(h)}\!(K)$$

\vspace{-2ex}
\noindent
where $M=\max_{(\alpha,\beta)\in \mathcal{E}} (|\alpha|+|\beta|)$ and $|\alpha|:=\alpha_1+\cdots+\alpha_n$. The non-commutative relation on $D_n^{(h)}\!(K)$ is defined by $\d_i x_i:=x_i \d_i+h^2$ to keep the result homogeneous.

The $K$-vector space isomorphism $\Psi$ is extended to be
\begin{eqnarray*}
\Psih:D_n^{(h)}\!(K) &\longrightarrow& K[h,\x,\xi]\\
h^{\gamma}\x^{\alpha}\d^{\beta} &\mapsto& h^{\gamma}\x^{\alpha}\xi^{\beta}.
\end{eqnarray*}
The set of all monomials and the set of all terms of $K[h,\x,\xi]$ are denoted by $\mathcal{M}_{\xi}^{(h)}$ and $\mathcal{T}_{\xi}^{(h)}$, respectively. For any $g=\sum_{(\gamma,\alpha,\beta)\in\mathcal{E}}\!c_{\gamma\alpha\beta}h^{\gamma}\x^{\alpha}\d^{\beta}$ for some $\mathcal{E}\subset \Z_{\geq 0}\times\Z_{\geq 0}^{n}\times\Z_{\geq 0}^{n}$, we define the set $\mon(g):=\{h^{\gamma}\x^{\alpha}\d^{\beta}\mid (\gamma,\alpha,\beta)\in\mathcal{E}, c_{\gamma\alpha\beta}\neq 0\}$ and the support $\supp(g)=\{(\gamma,\alpha,\beta)\in\mathcal{E}\mid c_{\gamma\alpha\beta}\neq 0\}$.
The tropical term order $\leq$ on $\mathcal{T}_{\xi}$ induces the tropical term order $\leq_h$ on $\mathcal{T}_{\xi}^{(h)}$ as follows.
\begin{defn}\label{deftermorderDnhvaluation}
For $ah^{\gamma}\x^{\alpha}\xi^{\beta},bh^{\gamma'}\x^{\alpha'}\d^{\beta'}\in\mathcal{T}_{\xi}^{(h)}$, we write:
\vspace{-0.75ex}
\begin{enumerate}[(1)]
\item $ah^{\gamma}\x^{\alpha}\xi^{\beta}<_h bh^{\gamma}\x^{\alpha'}\xi^{\beta'}$ if and only if:
\vspace{-0.5ex}
\begin{enumerate}[(i)]
\item $\deg(h^{\gamma}\x^{\alpha}\xi^{\beta})<\deg(h^{\gamma'}\x^{\alpha'}\xi^{\beta'})$ or
\item $\deg(\x^{\alpha}\d^{\beta}h^{\gamma})=\deg(\x^{\alpha'}\d^{\beta'}h^{\gamma'})$ \ and  \ $a\x^{\alpha}\d^{\beta}< b\x^{\alpha'}\d^{\beta'}$;
\end{enumerate}
\item $ah^{\gamma}\x^{\alpha}\xi^{\beta}=_hh^{\gamma'}b\x^{\alpha'}\xi^{\beta'}$  if and only if $\val(a)=\val(b)$ and $(\alpha,\beta,\gamma)=(\alpha',\beta',\gamma')$.
\end{enumerate}
\end{defn}
With respect to $<_h$, we define $\LTh(f):=\LTh(\Psih(f))$, $\LMh(f):=\LMh(\Psih(f))$, and $\LCh(f):=\LCh(\Psih(f))$ for any $f\in D_n^{(h)}\!(K)$.

In \cite{ADH2023}, we can compute a Gr\"{o}bner basis of a left ideal $I\subset D_n(K)$ by utilizing the homogenized Weyl algebra $D_n^{(h)}\!(K)$. Let $\mathcal{F}=\{\mathfrak{f}_1,\cdots,\mathfrak{f}_{\sn}\}$ be a generating set of $I$. We use the Buchberger algorithm with the input $\{H(\mathfrak{f}_1),\cdots,H(\mathfrak{f}_{\sn})\}$, where $H(\mathfrak{f}_i)$ is the homogenization of $\mathfrak{f}_i$. Once we obtain a Gr\"{o}bner basis $G^{(h)}=\{g_1^{(h)},\cdots,g_m^{(h)}\}$ in $D_n^{(h)}\!(K)$, we can get a Gr\"{o}bner basis $G$ in $D_n(K)$ by the dehomogenization process, i.e., $G:=\{g^{(h)}_{i_{\mid_{h=1}}}\mid i=1,\cdots,m\}$.

\section{Syzygies and Signatures}
\label{sec_syzsignature}
\vspace{-1ex}
From now on, the ring $D_n^{(h)}\!(K)$ is denoted by $R_n$ or just by $R$ if there is no ambiguity.
Let $\mathcal{F}=\{\mathfrak{f}_1,\cdots,\mathfrak{f}_{\sn}\}$ be a generating set of a left ideal $I$ in $R$, and write $\mathcal{F}_k:=\{\mathfrak{f}_1,\cdots,\mathfrak{f}_k\}$, $1\leq k\leq\mn$.
Let $G_k$ be a Gr\"{o}bner basis of the left ideal $I_k:=R\mathcal{F}_k$ in $R$.
In this paper, we are going to compute Gr\"{o}bner bases $G_i$'s increasingly by using an F5 algorithm. We start from the Gr\"{o}bner basis $G_1=\{\mathfrak{f}_1\}$ of $I_1$. We then compute a Gr\"{o}bner basis $G_k$ of $I_k$ by using some information of the previous computed Gr\"{o}bner bases $G_1,\cdots,G_{k-1}$, $2\leq k\leq\mn$.

We  assume that $\mathfrak{f}_k$'s are ordered increasingly by $\leq_h$ on $\LTh(\mathfrak{f}_k)$'s.
We also assume that for $2\leq j\leq \mn$, the generating set $\mathcal{F}_j$ is such that no term of $\mathfrak{f}_{j}$ is divisible by an element of $\LMh(G_{j-1})$. We can use the tropical division algorithm (Algorithm 1 in \cite{ADH2023}) to divide $\mathfrak{f}_j$ by $G_{j-1}$ if there is a term of $\Psih(\mathfrak{f}_j)$ which is divisible by an element of $\LMh(G_{j-1})$. 

\subsection{Syzygies}
Now, consider the module
$R^{\sn}:=\{(a_1,\cdots,a_{\sn})\mid a_i\in R, \ 1\leq i\leq \mn\}$
over $R$.
We know that $R^{\sn}=\bigoplus_{i=1}^{\sn} R \ve_i$ is a free module with the canonical basis $\{\ve_1,\cdots,\ve_{\sn}\}$. This implies that every $\a\in R^{\sn}$ can be uniquely written as $\a=\sum_{i=1}^{\sn} a_i\ve_i$ for some $a_1,\cdots,a_{\sn}\in R$. Moreover, it can be uniquely ({\it up to an order of the terms}) expressed as
$\a=\sum_{i=1}^{\sn}\sum_{(\gamma\alpha\beta)} c_{\gamma\alpha\beta}h^{\gamma}\x^{\alpha}\d^{\beta}\ve_i$
for some $c_{\gamma\alpha\beta}\in K^*$, and $h^{\gamma}\x^{\alpha}\d^{\beta}\in\mathcal{M}^{(h)}$, where $\mathcal{M}^{(h)}:=\{h^{\gamma}\x^{\alpha}\d^{\beta}\mid \gamma\in\Z_{\geq 0}, \alpha,\beta\in\Z_{\geq 0}^n\}\subset R$.
We define a $K$-vector space isomorphism
\vspace{-0.5ex}
\begin{eqnarray*}
\Phi:R^{\sn} &\longrightarrow& (K[h,\x,\xi])^{\sn}\\
h^{\gamma}\x^{\alpha}\d^{\beta}\ve_i &\mapsto& h^{\gamma}\x^{\alpha}\xi^{\beta}\ve_i
\end{eqnarray*}
and write $\mathbb{M}_{\xi}:=\{h^{\gamma}\x^{\alpha}\xi^{\beta}\ve_i  \mid  \gamma\in\Z_{\geq 0}, \alpha,\beta\in\Z_{\geq 0}^{n},  1\leq i\leq\mn\}$ and
$\mathbb{T}_{\xi}:=\{ch^{\gamma}\x^{\alpha}\xi^{\beta}\ve_i  \mid c\in K^*,  \gamma\in\Z_{\geq 0},   \alpha,\beta\in\Z_{\geq 0}^{n},  1\leq i\leq\mn\}$.

A module element $\g=(g_1,\cdots,g_{\sn})\in R^{\sn}$ is said to be \emph{homogeneous} if $g_i$ is homogeneous for all $1\leq i\leq\mn$. The total degrees of $g_i$'s need not be equal.

\begin{defn}
Let $\nu:R^{\sn}\longrightarrow I$ be the surjective $R$-module homomorphism given by $\ve_i\mapsto \mathfrak{f}_i$.
Any element $\bm{a}=(a_1,\cdots,a_{\sn})$ in $\ker(\nu)$ is called a \emph{syzygy} of $\mathcal{F}$.
This means that $\sum_{i=1}^{\sn}a_i \mathfrak{f}_i=0$. The set of all syzygies of $\mathcal{F}$ is denoted by $\syz{\mathcal{F}}$. Hence, $\syz{\mathcal{F}}=\ker(\nu)$.
\end{defn}

\begin{defn}
Let $\bm{a}=(a_1,\cdots,a_{\sn})$ be a syzygy of $\mathcal{F}$ and $i\in\{1,\cdots,\mn\}$ an index such that
$a_i\neq 0$ and $a_j=0$ for all $j>i$. The module monomial $\LMh(a_i)\ve_i\in \mathbb{M}_{\xi}$ is called the \emph{leading monomial} of the syzygy $\bm{a}$.
\end{defn}

We write $\lsyz{\mathcal{F}}$ the submodule of $(K[h,\x,\xi])^{\sn}$ generated by the leading monomials of syzygies of $\mathcal{F}$. The set ${\rm NS}(\syz{\mathcal{F}}):=\mathbb{M}_{\xi}\backslash\lsyz{\mathcal{F}}$ is called the \emph{normal set} of the syzygies of $\mathcal{F}$.

\begin{eg}
To provide a simple example, let us consider the polynomial ring $\Q[x,y,z]$ over the field $\Q$ with $2$-adic valuation. We use the tropical term order $\leq_h$ (restricted to $\Q[x,y,z]$) defined by the weight vectors $w=(1,1,1)$, $\omega=(-1,-1,-1)$, and the lexicographical order $\prec$ with $x\succ y\succ z$. Let $\mathcal{F}=\{\mathfrak{f}_1=z,\mathfrak{f}_2=y,\mathfrak{f}_3=x^2\}\subset \Q[x,y,z]$.
\begin{itemize}
\vspace{-1ex}
\item The module element $\bm{a}=(-x^2+y,-x^2-z,y+z)\in (\Q[x,y,z])^3$ is a non-homogeneous syzygy since $\nu(a)=0$ and $a_1=-x^2+y$ is non-homogeneous. The leading monomial of the syzygy $\bm{a}$ is $y\ve_3$.
\vspace{-1ex}
\item The syzygy $\bm{b}=-x^2\ve_2+y\ve_3$ is a homogeneous syzygy since $\nu(\bm{b})=0$ and both $-x^2$ and $y$ are homogeneous. The leading monomial of the syzygy $\bm{b}$ is $y\ve_3$.
\end{itemize}
\end{eg}

\subsection{Signatures}
In the classical signature-based algorithm (see: \cite{EF2017} for example), the (minimal) signature of a polynomial $f
\in K[\x]$ is defined as the minimal element of the set $\{\LM_{<_m}(\f)\mid \f\in(K[\x])^{\sn},  \nu(\f)=f\}$ for some total order $\leq_m$ on the set of module monomials of $(K[\x])^{\sn}$. Usually, the total order $\leq_m$ on $(K[\x])^{\sn}$ needs to be compatible to the corresponding monomial order on $K[\x]$, namely $\preceq_m$, in the sense that $m_1\preceq_m m_2$ if and only if $m_1\ve_i\leq_m m_2\ve_i$ for any monomials $m_1$ and $m_2$ in $K[\x]$ (see: Convention 2.1\! (a) in \cite{EF2017}).

Following \cite{Vaccon2021}, we are going to define the signature of an operator $f\in R$ to be a module element of $\mathbb{M}_{\xi}$ instead of a module term of $\mathbb{T}_{\xi}$. Therefore, a total order on $\mathbb{M}_{\xi}$ is needed.

\begin{defn}\label{def_totalordermodule}
Recall $\preceq$ a monomial order on $\mathcal{M}_{\xi}$ from Definition \ref{deftroptermorderDn}.\\
Let  $h^{\gamma}\x^{\alpha}\xi^{\beta}\ve_i, h^{\gamma'}\x^{\alpha'}\xi^{\beta'}\ve_j\in\mathbb{M}_{\xi}$.
We write $h^{\gamma}\x^{\alpha}\xi^{\beta}\ve_i\leq_{\rm sign} h^{\gamma'}\x^{\alpha'}\xi^{\beta'}\ve_j$ if:
\vspace{-1ex}
\begin{enumerate}[(1)]
\item $i<j$, or
\vspace{-1ex}
\item $i=j$ and $\deg(h^{\gamma}\x^{\alpha}\xi^{\beta})< \deg(h^{\gamma'}\x^{\alpha'}\xi^{\beta'})$, or
\vspace{-1ex}
\item $i=j$ and $\deg(h^{\gamma}\x^{\alpha}\xi^{\beta})= \deg(h^{\gamma'}\x^{\alpha'}\xi^{\beta'})$ and $w\cdot (\alpha,\beta)<w\cdot (\alpha',\beta')$, or
\vspace{-1ex}
\item $i=j$ and $\deg(h^{\gamma}\x^{\alpha}\xi^{\beta})= \deg(h^{\gamma'}\x^{\alpha'}\xi^{\beta'})$ and $w\cdot (\alpha,\beta)=w\cdot (\alpha',\beta')$ and:
\vspace{-1ex}
\begin{enumerate}[(a)]
\item $h^{\gamma}\x^{\alpha}\xi^{\beta}\ve_i\notin \lsyz{\mathcal{F}}$ and $h^{\gamma'}\x^{\alpha'}\xi^{\beta'}\ve_i\in \lsyz{\mathcal{F}}$, or
\item $h^{\gamma}\x^{\alpha}\xi^{\beta}\ve_i,h^{\gamma'}\x^{\alpha'}\xi^{\beta'}\ve_i\in \lsyz{\mathcal{F}}$ and $\x^{\alpha}\xi^{\beta}\preceq \x^{\alpha'}\xi^{\beta'}$, or
\item $h^{\gamma}\x^{\alpha}\xi^{\beta}\ve_i,h^{\gamma'}\x^{\alpha'}\xi^{\beta'}\ve_i\notin \lsyz{\mathcal{F}}$ and $\x^{\alpha}\xi^{\beta}\preceq \x^{\alpha'}\xi^{\beta'}$.
\end{enumerate}
\end{enumerate}
\end{defn}

The clauses (2), (3), and (4) in Definition \ref{def_totalordermodule} are designed as similar as possible to Definition \ref{deftermorderDnhvaluation}. In detail, the clause (2) and (3) correspond to the clause (i) in Definition \ref{deftermorderDnhvaluation} and the clause (i) in Definition \ref{deftroptermorderDn}, respectively. In both Definition \ref{def_totalordermodule} and Definition \ref{deftermorderDnhvaluation}, we use the same monomial order $\prec$ as a tie-breaker.
We recall a remark from \cite{Vaccon2021} with some adjustments for the context in this paper.
\vspace{-0.5ex}
\begin{rmk}
The order $<_{\rm sign}$ defined in Definition \ref{def_totalordermodule} may not satisfy Convention 2.1\! (a) in \cite{EF2017}: it may happen that $h^{\gamma}\x^{\alpha}\xi^{\beta}<_h h^{\gamma'}\x^{\alpha'}\xi^{\beta'}$, and yet that $h^{\gamma}\x^{\alpha}\xi^{\beta}\ve_i>_{\rm sign} h^{\gamma'}\x^{\alpha'}\xi^{\beta'}\ve_i$. This condition will happen if $h^{\gamma'}\x^{\alpha'}\xi^{\beta'}\ve_i\notin \lsyz{\mathcal{F}}$ and $h^{\gamma}\x^{\alpha}\xi^{\beta}\ve_i\in \lsyz{\mathcal{F}}$, or if both lie in or out of $\lsyz{\mathcal{F}}$ and $\x^{\alpha'}\xi^{\beta'}\prec\x^{\alpha}\xi^{\beta}$.
\end{rmk}

We identify some special parts of a module element $\f\in R^{\sn}$. For $\f\in R^{\sn}$, we define $\LT_{\rm sign}(\f)$ to be the biggest module term of $\Phi(\f)$. Moreover, we define $\LM_{\rm sign}(\f)$ to be the module monomial corresponding to $\LT_{\rm sign}(\f)$ and $\LC_{\rm sign}(\f)$ the corresponding coefficient. Here, the subscript ``sign'' indicates the use of the total order $\leq_{\rm sign}$.

\vspace{-0.75ex}
\begin{prop}\label{prop_LMofSyz}
Let $\bm{s}=a\sigma\ve_i+g'\ve_i+\sum_{j=1}^{i-1} g_j\ve_j$ be a syzygy of $\mathcal{F}$ with leading monomial $\Psih(\sigma)\ve_i$ and where $\sigma\notin\mon(g')$. Then,
$\LM_{\rm sign}(\bm{s})=\Psih(\sigma)\ve_i=\LMh(a\sigma+g')\ve_i.$
\end{prop}
\begin{proof}
It is clear $\LMh(a\sigma+g')\ve_i=\Psih(\sigma)\ve_i$, by definition of the leading monomial of a syzygy. This means that $\LTh(g')<_h \Psih(a\sigma)$.
\vspace{-0.7ex}

\noindent
By Definition \ref{def_totalordermodule}, $\LM_{\rm sign}(\bm{s})=\LM_{\rm sign}((a\sigma+g')\ve_i)$. If $\LTh(g')<_h \Psih(a\sigma)$ occurs because of total degree (Definition \ref{deftermorderDnhvaluation}) or of the clause (i) in Definition \ref{deftroptermorderDn}, then clearly we have $\LM_{\rm sign}(\bm{s})=\LM_{\rm sign}((a\sigma+g')\ve_i)=\LMh(a\sigma+g')\ve_i$. Now, assume that $\LTh(g')<_h \Psih(a\sigma)$ occurs due to the clause (ii) or (iii) in Definition \ref{deftroptermorderDn}. If all module terms of $\Phi(g'\ve_i)$ belong to ${\rm NS}(\syz{\mathcal{F}})$, then by the fact that $\Psih(\sigma)\ve_i\in\lsyz{\mathcal{F}}$ and by Definition \ref{def_totalordermodule} (4)(a), we get $\LM_{\rm sign}((a\sigma+g')\ve_i)=\Psih(\sigma)\ve_i$. If there is a module term of $\Phi(g'\ve_i)$ lies in $\lsyz{\mathcal{F}}$, then we must have $\LM_{\rm sign}((a\sigma+g')\ve_i)=\Psih(\sigma)\ve_i$ because both $<_{\rm sign}$ and $<$ use the same tie-breaker, i.e., $\prec$.
\end{proof}

\vspace{-1ex}
\begin{defn}
A module element $\f\in R^{\sn}$ with $\nu(\f)\neq 0$ is said to be \emph{clean} if the total symbol $\Phi(\f)$ has no terms belonging to $\lsyz{\mathcal{F}}$.
\end{defn}

\vspace{-0.5ex}
In the proof of Proposition \ref{prop_simple}, we reduce iteratively an element of $R^{\sn}$ by some syzygies until the corresponding reduction result has no terms belonging to $\lsyz{\mathcal{F}}$. We worry about the termination of these reductions due to the involvement of the valuations of coefficients on the term order $\leq_h$. To ensure that the reductions terminate in finite steps, we provide Algorithm \ref{alg_Red_bysyz} to explain our reductions algorithmically. This algorithm is not a part of our main algorithm (the F5 algorithm) in this paper; rather, it is only used to demonstrate that we can reduce a module element by a subset of $\syz{\mathcal{F}}$. In this algorithm, we use \'{e}cart function $E(f,g)=\#\left({\rm supp}(g)\backslash{\rm supp}(f)\right)$, where $f,g\in R$.

\vspace{0ex}
\begin{algorithm}[H]
\caption{Reduction of a module element by syzygies \qquad --- \qquad \textsc{RedBySyz}$(\bm{g},i)$}\label{alg_Red_bysyz}
\KwInput{A homogeneous $\g=(g_1,\cdots,g_{\sn})\in R^{\sn}$;\quad and \ \ an index $i$, $1\leq i\leq \mn$, such that $g_i\neq 0$.
} 

\KwOutput{$\g'=(g_1',\cdots,g_{\sn}')\in R^{\sn}$ such that $\nu(\g')=\nu(\g)$ and no module term of $\Phi(g_i'\ve_i)$ lies in ${\rm LSyz}(\mathcal{F})$
}

$S=\{\bm{s}=(s_1^{(1)}\!\!,\cdots\!,s_{i}^{(1)},0,\cdots\!,0)\in \syz{\mathcal{F}} \ \mid \bm{s} \text{ is homogeneous}, \deg(s_{i})\!=\!\deg(g_i), \LCh(s_i^{(l)})\!=\!1\}$\\
$T\leftarrow S$\label{line_RedBySyz_initstart}\\
$P\leftarrow \emptyset$\\
$\q^{(0)}\leftarrow \g-(0,\cdots,0,g_{i+1},\cdots,g_{\sn})$\label{line_RedBySyz_prepg}\\
$\r^{(0)}\leftarrow (0,\cdots,0,g_{i+1},\cdots,g_{\sn})$\\
$j\leftarrow 0$;\quad 
$a\leftarrow 0$\label{line_RedBySyz_initend}\\

\SetAlgoLined
\SetKwBlock{Begin}{While $q_i^{(j)}\neq 0$ do}{}
\Begin{
    \SetAlgoLined
    \SetKwBlock{Begin}{If {\normalfont  there is no $\bm{t}\!=\!(t_1,\cdots\!,t_{i},0,\cdots\!,0)\!\in\! T$ with $\LMh\!(t_i)\ve_i\!=\!\LMh\!(q_i^{(j)})\ve_i$} then}{}
    \Begin{
      $\r^{(j+1)}\leftarrow \r^{(j)}+\Psih^{-1}(\LTh(q_i^{(j)}))\ve_i$\label{line_RedBySyz_movetor}\\
      $\q^{(j+1)}\leftarrow \q^{(j)}-\Psih^{-1}(\LTh(q_i^{(j)}))\ve_i$\label{line_RedBySyz_takefromq}\\
      $h_{l,j+1}\leftarrow h_{l,j}$ \ \ for all $1\leq l\leq a$\\
      $T\leftarrow T\cup\{\q^{(j)}\}$
    }
    \SetAlgoLined
    \SetKwBlock{Begin}{else}{}
    \Begin{
      {\bf Choose} $\bm{t}\!=\!(t_1,\cdots\!,t_{i},0,\cdots\!,0)\!\in\! T$ with $\LMh(t_i)\ve_i\!=\!\LMh(q_i^{(j)})\ve_i$ and with $E(q_i^{(j)},t_i)$ minimal among all such choices.\label{line_RedBySyz_choose}
      \BlankLine
      \SetAlgoVlined
      \SetKwBlock{Begin}{If {\normalfont $E(q_i^{(j)},t_i)>0$} then}{}
      \Begin{
        $T\leftarrow T\cup\{\q^{(j)}\}$
      }
    }
}
\BlankLine
\end{algorithm}

\begin{algorithm}[H]
\setcounter{AlgoLine}{17}
\vspace{-1ex}\nonl
\SetAlgoVlined
\SetKwBlock{Begin}{}{}
\Begin{\nonl
		\vspace{-2ex}
    \SetAlgoVlined
    \SetKwBlock{Begin}{}{}
    \Begin{   
      $c=\LCh(q_i^{(j)})\cdot (\LCh(t_i))^{-1}$\\
      $\q'\leftarrow \q^{(j)}-c\,\bm{t}$\label{line_RedBySyz_reduction}\\
        \SetAlgoVlined
      \SetKwBlock{Begin}{If {\normalfont $\bm{t}\in S$} then}{}
      \Begin{
        $\q^{(j+1)}\leftarrow \q'$\\
        \SetAlgoLined
        \SetKwBlock{Begin}{If {\normalfont $\bm{t}=\bm{s}^{(k)}\in P$ for some $k$} then}{}
        \Begin{
            $h_{k,j+1}\leftarrow h_{k,j} + c$\\
            $h_{l,j+1}\leftarrow h_{l,j}$ \ \ for all $1\leq l\leq a$, $l\neq k$\\
        }
        \SetAlgoVlined
        \SetKwBlock{Begin}{else}{}
        \Begin{
            $a\leftarrow a+1$\\
            $\bm{s}^{(a)}\leftarrow \bm{t}$\\
            $P\leftarrow P\cup\{\bm{s}^{(a)}\}$\\
            $h_{a,l}\leftarrow 0$ \ \ for all $0\leq l\leq j$\\
            $h_{a,j+1}\leftarrow c$\\
            $h_{l,j+1}\leftarrow h_{l,j}$ \ \ for all $1\leq l< a$\\
        }
        $\r^{(j+1)}\leftarrow \r^{(j)}$
      }
      \BlankLine
        \SetKwBlock{Begin}{If {\normalfont $\bm{t}$ was added to $T$ at some previous iteration, namely $\bm{t}=\bm{q}^{(k)}$ for some $k<j$,} then}{}
        \Begin{
          $\q^{(j+1)}\leftarrow (1-c)^{-1}\q'$\\
          $h_{l,j+1}\leftarrow (1-c)^{-1}(h_{l,j}-c\,h_{l,k})$ \ \ for all $1\leq l\leq a$\\
          $\r^{(j+1)}\leftarrow (1-c)^{-1}(\r^{(j)}-c\,\r^{(k)})$\label{line_RedBySyz_subtr}
        }
    }
    $j=j+1$
}
{\bf Return} $\g'=\q^{(j)}+\r^{(j)}$\label{line_RedBySyz_return}
\BlankLine
\end{algorithm}

\medskip
\begin{prop}
The output of Algorithm \ref{alg_Red_bysyz} is correct, and the algorithm terminates after a finite number of steps.
\end{prop}
\begin{proof}
({\it Correctness proof}) We are going to show that the following properties hold at each step of the algorithm.
\begin{enumerate}[(i)]
\vspace{-0.5ex}
\item $\g=\q^{(j)}+\left(\sum_{\bm{s}^{(l)}\in P} h_{l,j} \bm{s}^{(l)}\right) +\r^{(j)}$
\vspace{-0.5ex}
\item no module term of $\Phi(r_i^{(j)}\ve_i)$ belongs to $\lsyz{\mathcal{F}}$
\vspace{-0.5ex}
\item $\LTh(q_i^{(j+1)})<_h\LTh(q_i^{(j)})$ if $q_i^{(j+1)}\neq 0$.
\end{enumerate}
At the initial stage (Line \ref{line_RedBySyz_initstart} -- \ref{line_RedBySyz_initend}), Property (i) and Property (ii) hold. Now, we are going to show that all properties above continue to hold at each step. We devide the proof into three possibilities of the steps.
\begin{itemize}
\item Assume that $\LTh(q_i^{(j)})\notin\lsyz{\mathcal{F}}$. For this case, we can see in Line \ref{line_RedBySyz_movetor} -- \ref{line_RedBySyz_takefromq} that we only move $\Psih^{-1}(\LTh(q_i^{(j)}))\ve_i$ from $\q^{(j)}$ to $\r^{(j+1)}$ so that we have $\q^{(j+1)}+\r^{(j+1)}=\q^{(j)}+\r^{(j)}$. Hence, Property (i) holds for $j+1$. Since $\LTh(q_i^{(j)})\ve_i$ lies out of $\lsyz{\mathcal{F}}$ and $\Psih^{-1}(\LTh(q_i^{(j)}))\ve_i$ is added to $\r^{(j+1)}$, Property (ii) holds. Property (iii) also holds since $\LTh(q_i^{(j)})$ is non-leading term of $q_i^{(j+1)}$.

\item Assume that there exists $\bm{t}\in S\subseteq T$ such that the condition in Line \ref{line_RedBySyz_choose} holds. Write $\bm{t}=\bm{s}^{(k)}$ for some $1\leq k\leq a$. The value of $a$ depends on the membership of $\bm{t}$ --- either $\bm{t}\in P$ or $\bm{t}\notin P$.  From the algorithm, we have $\q^{(j+1)}=\q^{(j)}-c\,\bm{s}^{(k)}$, and $h_{k,j+1}=h_{k,j}+c$. Then,
\begin{align*}
\q^{(j+1)}+ h_{k,j+1}\bm{s}^{(k)} \ &= \ \q^{(j)}-c\,\bm{s}^{(j)} \ \ + \ \ h_{k,j}\bm{s}^{(k)}+c\,\bm{s}^{(k)}\\
&= \ \q^{(j)} + h_{k,j}\bm{s}^{(k)}.
\end{align*}
For $l\neq k$, we have $h_{k,j+1}=h_{k,j}$, and hence, $\q^{(j+1)}+ h_{l,j+1}\bm{s}^{(l)}=\q^{(j)} + h_{l,j}\bm{s}^{(l)}$. Thus, Property (i) holds. Property (ii) also holds since $\r^{(j+1)}=\r^{(j)}$.
Since $q_i^{(j+1)}\ve_i$ is obtained from $q_i^{(j)}\ve_i$ by cancelling $\Psih^{-1}(\LTh(q_i^{(j)}))\ve_i$ (see Line \ref{line_RedBySyz_reduction}), we get $\LTh(q_i^{(j+1)})<_h \LTh(q_i^{(j)})$.

\item Now, assume that there exists $\bm{t}=\q^{(k)}\in T$ for some $k<j$ such that the condition in Line \ref{line_RedBySyz_choose} holds.
Here, we have $\LMh(q_i^{(k)})\ve_i=\LMh(q_i^{(j)})\ve_i$. Let $\LTh(q_i^{(k)})=c_k h^{\delta}\x^{\alpha}\xi^{\beta}$ and $\LTh(q_i^{(j)})=c_j h^{\delta}\x^{\alpha}\xi^{\beta}$. Since Property (iii) holds for the stages $\leq j$, we get $-\val(c_j)+\omega\cdot (\alpha,\beta)<-\val(c_k)+\omega\cdot (\alpha,\beta)$. Then, we get $\val(c_j)>\val(c_k)$, and hence, $\val(c)=\val(c_j\cdot c_k^{-1})>0$. As a consequence, we get $1-c\neq 0$.\\
Note that we have $\g=\q^{(j)}+\left(\sum_{\bm{s}^{(l)}\in P_{(j)}} h_{l,j} \bm{s}^{(l)}\right) +\r^{(j)}$ and $\g=\q^{(k)}+\left(\sum_{\bm{s}^{(l)}\in P_{(k)}}^m h_{l,k} \bm{s}^{(l)}\right) +\r^{(k)}$, where $P_{(k)}$ indicates the set $P$ at stage $k$. Since $k<j$, we have $P_{(k)}\subseteq P_{(j)}$. Then, for $j+1$, we get

\vspace{-2.5ex}
{\footnotesize
\begin{align*}
\q^{(j+1)} &= (1-c)^{-1}\!(\q^{(j)}-c\,\q^{(k)})\\
    &= (1-c)^{-1}\! \left[\! \left(\!\g-\left(\sum_{\bm{s}^{(l)}\in P} h_{l,j} \bm{s}^{(l)}\!\right) -\r^{(j)}\right)-c\left(\g-\left(\!\sum_{\bm{s}^{(l)}\in P} h_{l,k} \bm{s}^{(l)}\right) -\r^{(k)}\!\right)\!\right]\\
    &= \g - \left(\sum_{\bm{s}^{(l)}\in P} (1-c)^{-1}(h_{l,j}-c\,h_{l,k})\bm{s}^{(l)}\right)-(1-c)^{-1}(\r^{(j)}-c\,\r^{(k)})\\
    &= \g - \left(\sum_{\bm{s}^{(l)}\in P} h_{l,j+1}\bm{s}^{(l)}\right)-\r^{(j+1)}
\end{align*}
}
\vspace{-2ex}

\noindent
Therefore, Property (i) holds. Since no module term of both $\Phi(r_i^{(j)}\ve_i)$ and $\Phi(r_i^{(k)}\ve_i)$ belongs to $\lsyz{\mathcal{F}}$, Property (ii) clearly holds. By the fact that $\val(-c)>0=\val(1)$, we obtain $\val(1-c)=\min(\val(1),\val(-c))=\val(1)=0$. Then, $\LTh(q_i^{(j+1)})=_h\LTh((1-c)^{-1}q_i')<_h\LTh(q_i^{(j)})$.
\end{itemize}
Now, look at Line \ref{line_RedBySyz_return} of the algorithm.  The $i^{\rm th}$ entry of $\q^{(j)}$ is zero, and all module terms of $\Phi(r_i^{(j)}\ve_i)$ lies out of $\lsyz{\mathcal{F}}$. Thus, no module term of $\Phi(g_i'\ve_i)$ lies in $\lsyz{\mathcal{F}}$. By Property (i), we get
\vspace{-1ex}
$$
\nu(\g')=\nu(\q^{(j)}+\r^{(j)})=\nu(\g-\left(\sum_{\bm{s}^{(l)}\in P} h_{l,j+1}\bm{s}^{(l)}\right))=\nu(\g).
$$

\noindent
({\it Termination proof}) The algorithm will terminate when the condition $q_i^{(j')}=0$ holds, that is equivalently $\supp(q_i^{(j')})=\emptyset$. We need to prove that this condition occurs for some positive integer $j'$. Note that the algorithm only reduces the $i^{\rm th}$ entry of $\q^{(t)}$, i.e., $q_i^{(t)}$, by $T$ iteratively, and we start with $\q^{(0)}:=(g_1,\cdots,g_i,0,\cdots,0)$. Observe that this algorithm computes homogeneous operators $q_i^{(t)}$'s of the same degree, namely of degree $d$. We know that there are only a finite number of possible supports of homogeneous $q\in R$ of total degree $d$. Therefore, after enough steps, the supports $\supp(q_i^{(t)})$'s for all the possible $\LMh(q_i^{(t)})$'s have been added to $T$ as the $i^{\rm th}$ entry of $\q^{(t)}$'s. After that, if the minimal of the \'{e}cart is non-zero, then some $\q^{(j)}$ is added to $T$. This contradicts the fact that there is some $\q^{(k)}$ added earlier with the same exact support of $q_i^{(k)}$ and the same leading monomial $\LMh(q_i^{(k)})$, which means that $E(q_i^{(j)},q_i^{(k)})=0$. Then, after that point, all the minimal \'{e}carts must be $0$, and we obtain the inclusion $\supp(q_i^{(j+1)})\subseteq\supp(q_i^{(j)})$. Moreover, we must have the strict inclusion $\supp(q_i^{(j+1)})\subsetneq\supp(q_i^{(j)})$ since we vanish $\LMh(q_i^{(j)})\ve_i$. As the consequence, the remaining steps are finite. We will have $\q^{(j')}$ with $q_i^{(j')}=0$, or equivalently $\supp(q_i^{(j')})=\emptyset$, for some positive integer $j'$. At this point, the process terminates.
\end{proof}

\vspace{0.5ex}
\begin{rmk}\label{rmk_redbysyz}
The output of \textsc{RedBySyz}$(\g,i)$, namely $\g'=(g_1',\cdots,g_{\sn}')$, satisfies $\LTh(g_i')\leq_h\LTh(g_i)$. It is because the supports of $g_i'$ come from $g_i$ (see Line \ref{line_RedBySyz_prepg} and Line \ref{line_RedBySyz_movetor} of Algorithm \ref{alg_Red_bysyz}) or from syzygies in $S$ used as reducer (see Lines \ref{line_RedBySyz_initstart},\ref{line_RedBySyz_choose},\ref{line_RedBySyz_reduction} of Algorithm \ref{alg_Red_bysyz}). The corresponding terms must be \mbox{$\leq_h\!\!\LTh(g_i)$;} and if a sum of two terms with the same support (say in the sum $r_i^{(j)}+(-c)r_i^{(k)}$, see Line \ref{line_RedBySyz_subtr}) occurs, then thanks to the sequences of Lemmata in (Section 3, \cite{ADH2023}), the corresponding result is \mbox{$\leq_h\!\!\max(\LTh(r_i^{(j)}),r_i^{(k)})$}. A multiplication by $(1-c)^{-1}$ in Line \ref{line_RedBySyz_subtr} does not affect the ordering because $\val((1-c)^{-1})=0$. 
\end{rmk}

\begin{prop}\label{prop_simple}
Any module element $\g=(g_1,\cdots,g_{\sn})\in R^{\sn}$ with $\nu(\g)\neq 0$ can be carried to be a clean module element, namely $\g'$, such that $\nu(\g')=\nu(\g)$.
\end{prop}
\begin{proof}
If $\g$ is clean, then $\g'=\g$.\\
Now, assume that $\g$ is not clean. We also assume that $\g$ is homogeneous; otherwise, we can: write $\g=\sum_{k=1}^m \q^{(k)}$ for some homogeneous $\q^{k}\in R^{\sn}$ and integer $m>1$, do the procedure below for each $\q^{(k)}$ to get clean module elements $\q'^{(k)}$'s, and then add them up to obtain a clean module element $\g'=\sum_{k=1}^m \q'^{(k)}$.

\noindent
Since $\g$ is not clean, $\Phi(\g)$ has at least one term lies in $\lsyz{\mathcal{F}}$, namely $a\Psih(\sigma)\ve_i$ the biggest one, where $a\in K^*$ and $\Psih(\sigma)\in\mathcal{M}_{\xi}^{(h)}$. Applying Algorithm \ref{alg_Red_bysyz}, we obtain  $\g^{(1)}:=$\textsc{RedBySyz}$(\g,i)$ such that $\nu(\g^{(1)})=\nu(\g)$ and no module term of $\Phi(g_i^{(1)}\ve_i)$ lies in $\lsyz{\mathcal{F}}$. Next, we apply Algorithm \ref{alg_Red_bysyz} to $\g^{(1)}$ and with argument $i-1$ to obtain $\g^{(2)}:=$\textsc{RedBySyz}$(\g^{(1)},i-1)$. At this point, we have $\nu(\g^{(2)})=\nu(\g^{(1)})=\nu(\g)$; and for all $i-1\leq l\leq n$ with $g_{l}^{(2)}\neq 0$, no module term of $\Phi(g_{l}^{(2)}\ve_l)$ lies in $\lsyz{\mathcal{F}}$. We continue this procedure iteratively for $i-2,i-3,\cdots,1$. Finally, we obtain $\g^{(i)}:=$\textsc{RedBySyz}$(\g^{(i-1)},1)$ such that $\nu(\g^{(i)})=\nu(\g^{(i-1)})=\cdots=\nu(\g^{(1)})=\nu(\g)$, and for all $1\leq l\leq n$ with $g_{l}^{(i)}\neq 0$, no term of $\Psih(g_{l}^{(i)})$ lies in $\lsyz{\mathcal{F}}$. Thus, we get $\g':=\g^{(i)}$ as desired. 
\end{proof}
\vspace{-0.55ex}
\begin{cor}\label{cor_4}
Let $\bm{s}=\alpha\sigma\ve_i+f\ve_i+\sum_{j=1}^{i-1} g_j\ve_j$ be a syzygy of $\mathcal{F}$ with leading monomial $\Psih(\sigma)\ve_i$ and where $\sigma\notin \mon(f)$. Then, there exists a syzygy $\alpha\sigma\ve_i+\bm{g}'$ for some clean module element $\bm{g}'=g_i'\ve_i+\sum_{j=1}^{i-1} g_j'\ve_j\!\in\! R^{\sn}$, $\sigma\notin\mon(g_i')$, such that $\LTh(g_i')<_h \Psih(\alpha\sigma)$ and $\LM_{\rm sign}(\bm{g}')<_{\rm sign} \Phi(\sigma\ve_i)$.
\end{cor}
\begin{proof}
Write $\bm{g}=f\ve_i+\sum_{j=1}^{i-1} g_j\ve_j$. Since $\nu(\bm{s})=0$, then $\nu(\bm{g})\neq 0$. By Proposition \ref{prop_simple}, we get a clean module element $\bm{g}'$ such that $\nu(\bm{g}')=\nu(\bm{g})$. We observe again the proof of Proposition \ref{prop_simple} involving Algorithm \ref{alg_Red_bysyz} to obtain the clean module element $\g'=(g_1',\cdots,g_{\sn}')$. Looking at ``$\g^{(1)}=$\textsc{RedBySyz}$(\g,i)$,'' this reduction step is for obtaining a clean module element $g_i^{(1)}\ve_i$ from $f\ve_i$. Here, we get $g_i':=g_i^{(1)}$. By Remark \ref{rmk_redbysyz}, we have $\LTh(g_i^{(1)})\leq_h \LTh(f)$. Hence, we have $\LTh(g_i')\leq_h \LTh(f)<_h \Psih(\alpha\sigma)$.
The proof that $\LM_{\rm sign}(\bm{g}')<_{\rm sign}\Phi(\sigma\ve_i)$ is straightforward because $\alpha\sigma\ve_i+\bm{g}'$ is a syzygy, and $\bm{g}'$ is clean; so that $\alpha\sigma\ve_i$ is the only possible module term belonging to $\lsyz{\mathcal{F}}$.
\end{proof}

What we can understand from Corollary \ref{cor_4} is that for every syzygy $\bm{s}$ with leading monomial $\Psih(\sigma)\ve_i$, we can always find its corresponding syzygy, say $\bm{s}'=\alpha\sigma\ve_i+\bm{g}'$ with leading monomial $\Psih(\sigma)\ve_i$, such that $\bm{g}'$ is clean and $\LM_{\rm sign}(\bm{s}')=\Psih(\sigma)\ve_i$. The last equation, $\LM_{\rm sign}(\bm{s}')=\Psih(\sigma)\ve_i$, is consistent with what we have in Proposition \ref{prop_LMofSyz}.

Let $f\in I\backslash\{0\}$ and $\f\in R^{\sn}$ be such that $\nu(\f)=f$. The leading monomial $\LM_{\rm sign}(\f)$ is called a \emph{guessed signature} of $f$. The gussed signature of $f$ is denoted by $\mathcal{S}(f)$. The guessed signature of $f$ is not unique.

Now, we are ready to define signatures.
\vspace{-0.45ex}
\begin{defn}\label{def_signature}
For $f\in I\backslash\{0\}$, we define the \emph{signature} of $f$ to be
\vspace{-1ex}
$$
\s(f):=\min_{w.r.t. \leq_{\rm sign}} \{\LM_{\rm sign}(\g)\mid \g\in R^{\sn} \text{ such that } \nu(\g)=f \}
$$
\end{defn}

\begin{rmk}\label{rmk_sigf_i}
In practice --- at the initialization part of Algorithm \ref{alg_mainF5}, for all $i\in\{1,\cdots,\mn\}$, we set $\s(\mathfrak{f}_i)=\ve_i$. The module element $\ve_i$ is valid as the signature of $\mathfrak{f}_i$ because $\nu(\ve_i)=\mathfrak{f}_i$ and we have the assumption for $\mathcal{F}$ at the second paragraph of Section \ref{sec_syzsignature}
\end{rmk}

From Definition \ref{def_signature}, the signature of $f$ is obtained by taking the minimal of leading module monomials of such the $\g$'s with respect to $\leq_{\rm sign}$ without considering coefficients. As a consequence, for any $c\in K^*$, we have $\s(f)=\s(c\,f)$ as provided in the following proposition.
\begin{prop}
Let $f\in I\backslash\{0\}$ and $c\in K^*$. Then, $\s(c\,f)=\s(f)$. 
\end{prop}
\begin{proof}
Let $\s(f)=\Phi(\sigma\ve_i)$ for some $\sigma\in\mathcal{M}^{(h)}_{\xi}$. This means that there exists $\f\in R^{\sn}$ such that: $\LM_{\rm sign}(\f)=\Phi(\sigma\ve_i)$, $\nu(\f)=f$, and there is no $\f'\in R^{\sn}$ with $\LM_{\rm sign}(\f')<_{\rm sign}\Phi(\sigma\ve_i)$ such that $\nu(\f')=f$. On the other hand, let $\s(c\,f)=\Phi(\tau\ve_k)$. This means that there exists $\g\in R^{\sn}$ such that: $\LM_{\rm sign}(\g)=\Phi(\tau\ve_k)$, $\nu(\g)=c\,f$, and there is no $\g'\in R^{\sn}$ with $\LM_{\rm sign}(\g')<_{\rm sign}\Phi(\tau\ve_k)$ such that $\nu(\g')=c\,f$. Suppose that $\Phi(\sigma\ve_i)\neq\Phi(\tau\ve_k)$.
\par\noindent
Case-1: if $\Phi(\tau\ve_k)<_{\rm sign}\Phi(\sigma\ve_i)$. From $\nu(\g)=c\,f$ and $\LM_{\rm sign}(\g)=\Phi(\tau\ve_k)$, we get $\nu(c^{-1}\g)=f$ and $\LM_{\rm sign}(c^{-1}\g)=\Phi(\tau\ve_k)$. This contradicts the minimality of $\s(f)$.
\par\noindent
Case-2: if $\Phi(\sigma\ve_i)<_{\rm sign}\Phi(\tau\ve_k)$. From $\nu(\f)=f$ and $\LM_{\rm sign}(\f)=\Phi(\sigma\ve_i)$, we get $\nu(c\,\f)=c\,f$ and $\LM_{\rm sign}(c\,\f)=\Phi(\sigma\ve_i)$. This contradicts the minimality of $\s(c\,f)$.
\end{proof}

\vspace{0ex}
\begin{prop}
The signature of $f\in I\backslash\{0\}$ lies in ${\rm NS}(\syz{\mathcal{F}})$.
\end{prop}
\begin{proof}
Let $\s(f)=\Psih(\sigma)\ve_i$ and $\f\in R^{\sn}$ be such that $\LM_{\rm sign}(\f)=\Psih(\sigma)\ve_i$ and $\nu(\f)=f$, and there is no $\q\in R^{\sn}$ with $\nu(\q)=f$ and $\LM_{\rm sign}(\q)<_{\rm sign}\Psih(\sigma)\ve_i$.
Suppose that $\Psih(\sigma)\ve_i$ belongs to $\lsyz{\mathcal{F}}$. Then, there exists $\g=\LC_{\rm sign}(\f)\sigma\ve_i+g'\ve_i+\sum_{j=1}^{i-1}r_j\ve_j\in R^{\sn}$ where $g'\in R$, $\sigma\notin\mon(g')$, $\LT_h(g')<_h \LC_{\rm sign}(\f)\Psih(\sigma)$, and $r_j\in R$, such that $\nu(\g)=0$. We have $\nu(\f-\g)=f$, and in the substraction $\f-\g$, the term $\LC_{\rm sign}(\f)\sigma\ve_i$ is vanished.

\noindent
If $\LT_h(g')<_h \LC_{\rm sign}(\f)\Psih(\sigma)$ occurs because of total degree (the clause (i) in Definition \ref{deftermorderDnhvaluation}) or $w$-weighted total degree (the clause (i) in Definition \ref{deftroptermorderDn}), then by Definition \ref{def_totalordermodule}, we get $\LM_{\rm sign}(\f-\g)<_{\rm sign}\Psih(\sigma)\ve_i$. This contradicts the minimality of $\s(f)$.
\par\noindent
Now, assume that $\LT_h(g')<_h \LC_{\rm sign}(\f)\Psih(\sigma)$ occurs because of the clause (ii) or (iii) in Definition \ref{deftroptermorderDn}. If $\f-\g$ is not clean, then by Proposition  \ref{prop_simple}, we can carry $\f-\g$ to be a clean module element $\f'$ such that $\nu(\f')=f$. By the assumption that $\Psih(\sigma)\ve_i\in\lsyz{\mathcal{F}}$ and by the definition of $\leq_{\rm sign}$, all terms of $\Phi(\f')$ are strictly smaller than $\Psih(\sigma)\ve_i$ with respect to $<_{\rm sign}$. Hence, $\LM_{\rm sign}(\f')<_{\rm sign}\Psih(\sigma)\ve_i=\s(f)$. This contradicts the minimality of $\s(f)$.
\end{proof}

\vspace{-0.5ex}
\begin{prop}\label{prop_criteria_tdotsigf}
Let $t\in\mathcal{T}^{(h)}$ and $f\in I\backslash\{0\}$. Then,
\vspace{-0.4ex}
\begin{align*}
\s(tf)<_{\rm sign}\Psih(t)\s(f) \quad&\Leftrightarrow\quad \Psih(t)\s(f)\in \lsyz{\mathcal{F}}\\
\s(tf)=\Psih(t)\s(f) \quad&\Leftrightarrow\quad \Psih(t)\s(f)\in{\rm NS}(\syz{\mathcal{F}}).
\end{align*}
\end{prop}
\begin{proof}
Let $\s(f)=\Psih(\sigma)\ve_i$. By the definition of signature, we have $f=\nu(\f)$ for some
$\f=a\sigma\ve_i+f'\ve_i+\sum_{j=1}^{i-1}q_j\ve_j\in R^{\sn},$
where $a\in K^*$, $f'\in R$ with $\LM_{\rm sign}(f'\ve_i)<_{\rm sign}\Psih(\sigma)\ve_i$, and $q_j\in R$; and there is no $\q\in R^{\sn}$ such that $\nu(\q)=f$ and $\LM_{\rm sign}(\q)<_{\rm sign} \Psih(\sigma)\ve_i$. Here, we have $\nu(t\f)=tf$ and $\LM_{\rm sign}(t\f)$ $=\Psih(t)\Psih(\sigma)\ve_i$.
\par\noindent
If $\Psih(t)\Psih(\sigma)\ve_i\in \lsyz{\mathcal{F}})$, then there exists $\g=at\sigma\ve_i+g'\ve_i+\sum_{j=1}^{i-1}r_j\ve_j\in R^{\sn}$ where $g'\in R$ with $\LTh(g')<_h\Psih(at\sigma)$, $r_j\in R$, and $\nu(\g)=0$. Since $\nu$ is homomorphism, we get $\nu(t\f-\g)=tf$. The module term $at\sigma\ve_i$ is vanished in the subtraction $t\f-\g$. If $t\f-\g$ is not clean, then we can carry it to be a clean module element $\f^*$ such that $\nu(\f^*)=tf$. By the hypothesis $\Psih(t)\Psih(\sigma)\ve_i\in \lsyz{\mathcal{F}})$ and by the definition of $\leq_{\rm sign}$, we have that all terms of $\Phi(\f^*)$ are strictly less than $\Psih(t)\Psih(\sigma)\ve_i$ with respect to $<_{\rm sign}$. Thus, we conclude that $\s(tf)<_{\rm sign} \Psih(t)\Psih(\sigma)\ve_i=\Psih(t)\s(f)$.
\par\noindent
Now, assume that $\s(tf)<_{\rm sign} \Psih(t)\s(f)$, and suppose that $\Psih(t)\s(f)\in{\rm NS}(\syz{\mathcal{F}})$. Let $\h$ be an element of $R^{\sn}$ such that $\nu(\h)=tf$, $\LM_{\rm sign}(\h)=\s(tf)$, and there is no $\q\in R^{\sn}$ with $\nu(\q)=tf$ and $\LM_{\rm sign}(\q)<_{\rm sign} \LM_{\rm sign}(\h)$. We may assume that both $t\f$ and $\h$ are clean; otherwise, bring it to be clean module elements by using Proposition \ref{prop_simple}. We have $\nu(t\f-\h)=tf-tf=0$, and hence, $t\f-\h$ is a syzygy. This contradicts the fact that $t\f-\h$ is clean.
\par\noindent
The second statement of the lemma follows the first statement.
\end{proof}

The one-to-one correspondence between $\LMh(I\backslash\{0\})$ and ${\rm NS}(\syz{\mathcal{F}})$ presented in Proposition \ref{prop_bijection} plays an important role in the next section. To prove this proposition, we need the following lemma.

\begin{lem}\label{lem_sig_af_bg}
Let $f,g\in I\backslash \{0\}$ be such that $\s(f)=\s(g)=\sigma\ve_i$ and $\LMh(f)\neq \LMh(g)$.
Then, there exist $a,b\in K^*$ such that
$
\s(af+bg)<_{\rm sign} \sigma\ve_i.
$
\end{lem}
\begin{proof}
Since $\s(f)=\s(g)=\sigma\ve_i$, there exist $\f,\g\in R^{\sn}$ such that $\nu(\f)=f$, $\nu(\g)=g$, and $\LM_{\rm sign}(\f)=\LM_{\rm sign}(\g)=\sigma\ve_i$ is minimal among all such module elements $\f$ and $\g$.
Set $a=\LC_{\rm sign}(\g)$ and $b=-\LC_{\rm sign}(\f)$. Then, $\nu(a\f+b\g)=af+bg\neq 0$ since $\LMh(f)\neq \LMh(g)$, and $\LT_{\rm sign}(a\f)+\LT_{\rm sign}(b\g)=\bm{0}$. We get $\LM_{\rm sign}(a\f+b\g)<_{\rm sign} \sigma\ve_i$, and hence, $\s(af+bg)<_{\rm sign}\sigma\ve_i$. 
\end{proof}

\vspace{0ex}
\begin{prop}\label{prop_bijection}
The function
\begin{eqnarray*}
\Omega: \LMh(I\backslash\{0\})&\longrightarrow& {\rm NS}(\syz{\mathcal{F}})\\
m&\mapsto&\min_{\leq_{\rm sign}}\; \{\s(f)\mid f\in I\backslash\{0\} \ \text{ with } \ \LMh(f)=m\}
\end{eqnarray*}
\vspace{-3ex}

\noindent
is bijective.
\end{prop}
\begin{proof}
Suppose that there exist $m_1,m_2\in\LMh(I\backslash\{0\})$ and $\sigma\ve_i\in {\rm NS}(\lsyz{\mathcal{F}})$ such that $m_1\neq m_2$ and $\Omega(m_1)=\sigma\ve_i=\Omega(m_2)$. Then, there exist $f,g\in I\backslash\{0\}$ such that $\LMh(f)=m_1$, $\LMh(g)=m_2$, and $\s(f)=\s(g)=\sigma\ve_i$ is minimal among all such $f$ and $g$. By Lemma \ref{lem_sig_af_bg}, there exist $a,b\in K^*$ such that $\s(af+bg)<_{\rm sign}\sigma\ve_k$. Note that $\LTh(af+bg)=\max(\LTh(af),\LTh(bg))$. Assume that $\LTh(af+bg)=\LTh(af)$. We have $\LMh(af)=m_1$. Thus, we have the operator $af+bg$ with $\LMh(af+bg)=m_1$ and $\Omega(m_1)<_{\rm sign}\sigma\ve_i$, contradicting the choice of such the operator $f$ with $\LM_h(f)=m_1$ such that $\s(f)$ is minimal.\\
Now, let $\tau\ve_j\in{\rm NS}(\syz{\mathcal{F}})$, and set
\vspace{-1ex}
$$
A=\{m\in\mathcal{M}_{\xi}^{(h)}\mid \exists g\in I\backslash\{0\}, \LMh(g)=m, \s(g)=\tau\ve_j\}.
$$
This set is not empty since $\LTh(\tau \mathfrak{f}_j)$ belongs to $A$. Write $t=\min(A)$ with respect to $\leq_h$. Then, we have an operator $g\in I\backslash\{0\}$ with $\LMh(g)=t$ and with $\s(g)=\tau\ve_j$. Suppose that $\Omega(t)<_{\rm sign}\tau\ve_j$. Then, there exists $g'\in I\backslash\{0\}$ such that $\LMh(g')=t$ and $\s(g')<_{\rm sign}\tau\ve_j$. Since $\LMh(g)=\LMh(g')=t$, there exist $a,b\in K^*$ such that $\LMh(ag+bg')<_h t$. Moreover, since $\s(g')<_{\rm sign}\s(g)=\tau\ve_j$, we have $\s(ag+bg')=\tau\ve_j$. Both facts $\LMh(ag+bg')<_h t$ and $\s(ag+bg')=\tau\ve_j$ contradict the minimality of $t$.
\end{proof}

\section{\texorpdfstring{$\s$}{}-Gr\"{o}bner bases}
\label{sec_sigGB}
\vspace{-1ex}
Before we discuss the theory of $\s$-Gr\"{o}bner bases, we need to understand the notion of $\s$-top reduction and $\s$-top-irreducibility.

\vspace{-1.5ex}
\subsection{\texorpdfstring{$\s$}{}-top-reduction}
\vspace{-1ex}
An $\s$-top-reduction of a non-zero $f\in I$ by a non-zero $g\in I$  with respect to a given monomial module $\sigma\ve_i$ is more specific than the classical reduction. When $\sigma\ve_i=\s(f)$, the $\s$-top-reductions preserve the signature from any changes during reductions. As consequence, we can keep track of where an element of $I$ come from.
\begin{defn}\label{def_stopred}
Let $f,g\in I\backslash\{0\}$, $q\in I$, and $\sigma\ve_i\in\mathbb{M}_{\xi}$. We say that $f$ $\s$-\emph{top-reduces} to $q$ by $g$ with respect to $\sigma\ve_i$ if there exists $t\in\mathcal{T}^{(h)}$ such that:
\vspace{-0.25ex}
\begin{enumerate}[(i)]
\item $\LMh(tg)=\LMh(f)$;
\vspace{-0.5ex}
\item $\s(tg)<_{\rm sign}\sigma\ve_i$.
\end{enumerate}
The operator $g$ is called $\s$-\emph{top-reducer} for $f$ with respect to $\sigma\ve_i$. In addition, if $\sigma\ve_i$ is not specified, then we mean that $\sigma\ve_i=\s(f)$.
\end{defn}

\vspace{-1ex}
\begin{defn}
We say that $f\in I$ is $\s$-\emph{top-irreducible} with respect to $\sigma\ve_i$ if $f=0$ or there is no $\s$-top-reducer $g\in I\backslash\{0\}$ for $f$ with respect to $\sigma\ve_i$.
\end{defn}

The concept of $\s$-top-irreducibility gives rise to the following theorem and its corollary.

\begin{thm}\label{thm_fzeroiff}
Let $f\in I$ be such that $f$ is $\s$-top-irreducible with respect to $\sigma\ve_i\in\mathbb{M}_{\xi}$, and $f$ is of the form $f=\nu(\;(\Psih^{-1}(a\sigma)+g)\ve_i+\sum_{j=1}^{i-1}q_j\ve_j)$ for some $a\in K^*$, $q_i,g\in R$ with $\LTh(g)<_h a\sigma$ and $\sigma\notin\mon(g)$. Then,
\vspace{-1ex}
$$
f=0\quad \text{if and only if}\quad \sigma\ve_i\in\lsyz{\mathcal{F}}.
$$
\vspace{-4ex}

\noindent
In addition, if $f\neq 0$, then $\sigma\ve_i=\s(f)=\Omega(\LMh(f))$.
\end{thm}
\begin{proof}
$(\Rightarrow).$ If $0=f=\nu(\;(\Psih^{-1}(a\sigma)+g)\;\ve_i+\sum_{j=1}^{i-1}q_j\ve_j)$, then  $\sigma\ve_i\in\lsyz{\mathcal{F}}$.
\par
\vspace{-1ex}

\noindent
$(\Leftarrow).$ Let $\sigma\ve_i\in\lsyz{\mathcal{F}}$. Suppose that $f\neq 0$. Let $m=\LMh(f)$ and $\tau\ve_k=\Omega(m)$. Here, we have $\Omega(m)<_{\rm sign}\sigma\ve_i$. From the fact that $\Omega(m)=\tau\ve_k$, there exists $q\in I$ such that $\LMh(q)=m$ and $\s(q)=\tau\ve_k$. Since $\tau\ve_k<_{\rm sign}\sigma\ve_i$, $q$ is an $\s$-top-reducer for $f$ with respect to $\sigma\ve_i$. This contradicts the hypothesis. Thus, we get $f=0$.
\par
\noindent
For the second statement, assume that $f\neq 0$. We claim that $\sigma\ve_i\in{\rm NS(\syz{\mathcal{F}})}$. To prove this claim, suppose that $\sigma\ve_i\in\lsyz{\mathcal{F}}$. Let $m=\LMh(f)$ and $\tau\ve_k=\Omega(m)$. Then, there exists $q\in I$ such that $\LMh(q)=m$ and $\s(q)=\tau\ve_k$. Here, $q$ is an $\s$-top-reducer for $f$ with respect to $\sigma\ve_i$, contradicting the hypothesis. Thus, we must have $\sigma\ve_i\in{\rm NS(\syz{\mathcal{F}})}$.\\
Since $\sigma\ve_i\in{\rm NS(\syz{\mathcal{F}})}$ and $f=\nu((\Psih^{-1}(a\sigma)+g)\ve_i+\sum_{j=1}^{i-1}q_j\ve_j)$, we get $\sigma\ve_i=\s(f)$. Now, suppose that $\sigma\ve_i\neq\Omega(\LMh(f))$. Then, we must have $\Omega(\LMh(f))<_{\rm sign}\sigma\ve_i$. There exists $q\in I$ such that $\LMh(q)=\LMh(f)$ and $\s(q)=\Omega(\LMh(q))$. Since $\Omega(\LMh(q))<_{\rm sign}\sigma\ve_i$, $q$ is an $\s$-top-reducer for $f$ with respect to $\sigma\ve_i$. This contradicts the hypothesis.
\end{proof}

\vspace{0ex}
\begin{cor}\label{cor_irriffomega}
A non-zero $f\in I$ is $\s$-top-irreducible with respect to $\s(f)$ if and only if $\s(f)=\Omega(\LMh(f))$.
\end{cor}
\begin{proof}
If $f\neq 0$ and $f$ is $\s$-top-irreducible with respect to $\s(f)$, then by the second statement of Theorem \ref{thm_fzeroiff}, we get $\s(f)=\Omega(\LMh(f))$. Conversely, if $\s(f)=\Omega(\LMh(f))$, then by the definition of $\Omega$, clearly $f$ is $\s$-top-irreducible with respect to $\s(f)$.
\end{proof}

For the $\s$-top-reduction in Definition \ref{def_stopred}, we can also say ``$f$ $\s$-\emph{top-reduces regularly} to $q$ by $g$ with respect to $\sigma\ve_i$.'' Another type of reduction is defined as follows.
\vspace{-1ex}
\begin{defn}\label{def_singularstopred}
Let $f,g\in I\backslash\{0\}$, $q\in I$, and $\sigma\ve_i\in\mathbb{M}_{\xi}$.
We say that $f$ $\s$-\emph{top-reduces singularly} to $q$ by $g$ with respect to $\sigma\ve_i$ if there exists $t\in\mathcal{T}^{(h)}$ such that:
\begin{enumerate}[(i)]
\vspace{-1ex}
\item $\LMh(tg)=\LMh(f)$;
\vspace{-1ex}
\item $\s(tg)=\sigma\ve_i$.
\end{enumerate}
\vspace{-0.5ex}

\noindent
The operator $g$ is called a \emph{singular} $\s$-\emph{top-reducer} for $f$ with respect to $\sigma\ve_i$. If $\sigma\ve_i$ is not specified, then we mean that $\sigma\ve_i=\s(f)$.
\end{defn}

Let $f,g\in I\backslash\{0\}$ and $q\in I$ be such that $f$ $\s$-top-reduces singularly to $q$ by $g$ with respect to $\s(f)$. It means that there exists $t\in\mathcal{T}^{(h)}$ such that the axioms (i) and (ii) in Definition \ref{def_singularstopred} hold. Let
$$\f=a\,\Phi^{-1}(\s(f))+(\text{lower module terms w.r.t.}<_{\rm sign})\quad \text{and}$$
$$\h=b\,\Phi^{-1}(\s(tg))+(\text{lower module terms w.r.t.}<_{\rm sign})$$
be elements of $R^{\sn}$ such that $a,b\in K^*$, $\nu(\f)=f$, $\nu(\h)=tg$, $\LM_{\rm sign}(\f)=\s(f)$, and $\LM_{\rm sign}(\h)=\s(tg)=\s(f)$.
In this reduction, we have $f-tg=q$, and its corresponding module representation is $\f-\h=\q$, where $\q\in R^{\sn}$ satisfies $\nu(\q)=q$. Since $\s(f)=\s(tg)$, we have either $\LM_{\rm sign}(\q)=\s(f)$ or $\LM_{\rm sign}(\q)<_{\rm sign}\s(f)$.
\begin{rmk}\label{rmk_preservesigornot}
Consider again what we have in the paragraph above.
\begin{enumerate}[(i)]
\vspace{-0.7ex}
\item The case $\LM_{\rm sign}(\q)=\s(f)$ occurs when $a-b\neq 0$. The meaningful of this case is that the reduction process preserves the signature.
\vspace{-0.7ex}
\item The case $\LM_{\rm sign}(\q)<_{\rm sign}\s(f)$ occurs when $a-b=0$. Such the reduction is not allowed in our computation because it will change the signature, and we are not able anymore to track of where the operator $q$ come from.
\end{enumerate}
\end{rmk}

\subsection{\texorpdfstring{$\s$}{}-Gr\"{o}bner bases}
In this sub-section, we present a theory of $\s$-Gr\"{o}bner basis. We are going to show that once we have an $\s$-Gr\"{o}bner basis, we automatically also have a Gr\"{o}bner basis.
\begin{defn}\label{def_sGB}
A non-empty subset $G\subset I$ is called an \emph{$\s$-Gr\"{o}bner basis} of $I$ if $RG=I$ and for every $\s$-top-irreducible $f\in I\backslash\{0\}$, there exists $g\in G$ and $m\in\mathcal{M}^{(h)}$ such that:
\vspace{-0.7ex}
\begin{enumerate}[(i)]
\item $\LMh(mg)=\LMh(f)$ and
\vspace{-0.5ex}
\item $\Psih(m)\s(g)=\s(f)$.
\end{enumerate}
\end{defn}

\noindent In particular, we have the following definition.
\begin{defn}\label{def_sGBupto}
A non-empty subset $G\subset I$ is called an \emph{$\s$-Gr\"{o}bner basis for $I$ in signature} $\sigma\ve_k$ (\emph{up to signature} $\sigma\ve_k$, resp.) if $RG=I_k$ and for every $\s$-top-irreducible $f\in I_k\backslash\{0\}$ with $\s(f)=\sigma\ve_k$ (with $\s(f)<_{\rm sign}\sigma\ve_k$, resp.), there exists $g\in G$ and $m\in\mathcal{M}^{(h)}$ such that:
\vspace{-0.7ex}
\begin{enumerate}[(i)]
\item $\LMh(mg)=\LMh(f)$ and
\vspace{-0.25ex}
\item $\Psih(m)\s(g)=\s(f)$.
\end{enumerate}
\end{defn}

The connection between Definition \ref{def_sGB} and Definition \ref{def_sGBupto} is that a non-empty subset $G\subset I$ is an $\s$-Gr\"{o}bner basis of $I$ if $G$ is an $\s$-Gr\"{o}bner basis for $I$ in all signatures.

\begin{prop}\label{prop_GBthenreducer}
If $G$ is an $\s$-Gr\"{o}bner basis, then for any $\s$-top-reducible $f\in I\backslash\{0\}$ with respect to $\s(f)$, there exist $g\in G$ and $m\in \mathcal{M}^{(h)}$ such that:
\vspace{-0.7ex}
\begin{enumerate}[(i)]
\item $\LMh(mg)=\LMh(f)$ and
\vspace{-0.5ex}
\item $\s(mg)=\Psih(m)\s(g)<_{\rm sign}\s(f)$.
\end{enumerate}
\end{prop}
\begin{proof}
Let $f\in I\backslash\{0\}$ be such that $\s$-top-reducible with respect to $\s(f)$. Then, there exists $q\in I\backslash\{0\}$ such that $\LMh(q)=\LMh(f)$ and $\s(q)<_{\rm sign}\s(f)$. We choose such the $q$ such that $\s(q)=\Omega(\LMh(q))$. By Corollary \ref{cor_irriffomega}, we get $q$ is $\s$-top-irreducible with respect to $\s(q)$. Since $G$ is an $\s$-Gr\"{o}bner basis, there exist $g\in G$ and $m\in\mathcal{M}^{(h)}$ such that $\LMh(mg)=\LMh(q)=\LMh(f)$ and $\s(mg)=\Psih(m)\s(g)=\s(q)<_{\rm sign}\s(f)$, as desired.
\end{proof}

\smallskip
The meaningful of Proposition \ref{prop_GBthenreducer} above is that any non-zero $\s$-top-reducible element of $I$ can be $\s$-top-reduced by an element of an $\s$-Gr\"{o}bner basis. This fact is important to prove the following proposition.

\begin{prop}
If $G$ is an $\s$-Gr\"{o}bner basis of $I$, then $G$ is a Gr\"{o}bner basis of $I$.
\end{prop}
\begin{proof}
By the definition of $\s$-Gr\"{o}bner basis, we have $RG=I$. Next, we need to prove that $\langle\LMh(G)\rangle=\langle\LMh(I)\rangle$.
Let $m\in\langle \LMh(I)\rangle$, and let $f\in I$ be such that $\LMh(f)=m$.
If $f$ is $\s$-top-reducible with respect to $\s(f)$, then by Proposition \ref{prop_GBthenreducer}, there exist $g\in G$ and $m_1\in \mathcal{M}^{(h)}$ such that $\LMh(m_1 g)=\LMh(f)$ and $\s(m_1 g)<_{\rm sign}\s(f)$.
If $f$ is $\s$-irredicible, then by the definition of $\s$-Gr\"{o}bner basis, there exists $g'\in G$ and $t\in\mathcal{M}^{(h)}$ such that $\LMh(tg')=\LMh(f)$ and $\Psih(t)\s(g')=\s(f)$. From two cases above, we conclude that $m=\LMh(f)$ lies in $\langle\LMh(G)\rangle$. Thus, we have $\langle\LMh(I)\rangle\subseteq\langle\LMh(G)\rangle$.
\end{proof}

\begin{prop}\label{propcontainfiniteGB}
Every $\s$-Gr\"{o}bner basis of $I$ contains a finite $\s$-Gr\"{o}bner basis of $I$.
\end{prop}
\begin{proof}
Let $G=\{g_i\mid i\in L\}$ be an $\s$-Gr\"{o}bner basis of $I$, where $L$ is an index set. Define a function
\vspace{-0.7ex}
\begin{eqnarray*}
\Lambda:G&\longrightarrow&\mathcal{M}_{\xi}^{(h)}\times\mathbb{M}_{\xi}\\
g_i&\mapsto&(\LMh(g_i),\s(g_i)).
\end{eqnarray*}
\vspace{-3.5ex}

\noindent
We can write $\Lambda(G)=\LMh(G)\times \{\s(g_i)\mid g_i\in G\}$, a subset of the monoid $\mathcal{M}^{(h)}_{\xi}\times \mathbb{M}_{\xi}$. Moreover, we know that $\mathcal{M}^{(h)}_{\xi}\times \mathbb{M}_{\xi}\cong (\mathcal{M}_{\xi}^{(h)})^{\sn+1}$, so the set $\Lambda(G)$ can be viewed as a subset of $(\mathcal{M}_{\xi}^{(h)})^{\sn+1}$. From the set $\Lambda(G)$, we can construct $M=(M_1,\cdots,M_{\sn+1})$, where $M_k$ is the monomial ideal generated by monomials in the $k^{\rm th}$ coordinate of $\Lambda(G)$, $1\leq k\leq \mn+1$. Here, we have $\Lambda(G)\subseteq M$.
By the Dickson lemma, there exists a finite subset $J\subset L$ such that monomials in the $k^{\rm th}$ coordinate of $\Lambda(\{g_j\mid j\in J\})$ generates $M_k$, $1\leq k\leq \mn+1$. Write $G_0:=\{g_j\mid j\in J\}$.

\noindent
We are going to prove that $G':=G_0\cup F$ is an $\s$-Groebner basis of $I$. We have taken the union with $\mathcal{F}$ to ensure that $G'$ generates the left ideal $I$. Let $f\in I$ be $\s$-top-irreducible with respect to $\s(f)$. Since $G$ is an $\s$-Gr\"{o}bner basis, there exist $g_i \in G$ and $m\in \mathcal{M}^{(h)}$ such that $\Psih(m)\LMh(g_i)=\LMh(mg_i)=\LMh(f)$ and $\s(mg_i)=\Psih(m)\s(g_i)=\s(f)$. If $i\in J$, then it is clear that $g_i\in G'$. Now, assume that $i\in L\backslash J$. Then, since $\Lambda(g_i)\in M$, there exist $j_i\in J$, $u\in \mathcal{M}^{(h)}_{\xi}$, and $v\ve_k\in\mathbb{M}_{\xi}$ such that
$(u,v\ve_k)\cdot \Lambda(g_{j_i})=\Lambda(g_i)=(\LMh(g_i),\s(g_i)).$
\begin{itemize}
\vspace{-0.75ex}
\item If $u=v$, then $m^*:=u\Psih(m)\in \mathcal{M}^{(h)}_{\xi}$ satisfies $m^*\LMh(g_{j_i})=\LMh(f)$ and $m^*\s(g_{j_i})=\s(f)$. This case satisfies the $\s$-Gr\"{o}bner basis axioms for $f$.
\vspace{-0.75ex}
\item If $u<_h v$, then $m':=u\Psih(m)\in \mathcal{M}^{(h)}_{\xi}$ satisfies $m'\LMh(g_{j_i})=\LMh(f)$ and $m'\s(g_{j_i})<_{\rm sign}\s(f)$. It contradicts the hypothesis that $f$ is $\s$-top-irreducible with respect to $\s(f)$.
\vspace{-0.75ex}
\item If $u>_h v$, then there exist $a\in K^*$ and $m':=v\Psih(m)\in\mathcal{M}_{\xi}^{(h)}$ such that the operator $q:=f-am'g_{j_i}\in R$ satisfies that $\LMh(q)=\LMh(f)$ and $\s(q)<_{\rm sign}\s(f)$. This contradicts the hypothesis that $f$ is $\s$-top-irreducible with respect to $\s(f)$.
\end{itemize}
Thus, we conclude that $G'$ is an $\s$-Groebner basis of $I$.
\end{proof}

\vspace{0ex}
\begin{defn}\label{def_primitive}
A non-zero $\s$-top-irreducible $g\in I$ is said to be \emph{primitive $\s$-top-irreducible} if there are no $\s$-top-irreducible $g'\in I\backslash\{0\}$ and $m'\in\mathcal{M}^{(h)}\backslash\{1\}$ such that:
\begin{enumerate}[(i)]
\vspace{-0.7ex}
\item $\LMh(m'g')=\LMh(g)$ and
\vspace{-0.5ex}
\item $\s(m'g')=\s(g)$.
\end{enumerate}
\end{defn}

Following the proof of Proposition \ref{propcontainfiniteGB}, in practice, we keep a subset of primitive $\s$-top-irreducible operators with different leading monomials in computing a finite $\s$-Gr\"{o}bner basis $G'$. We remark that in that proof, we set the finite $\s$-Gr\"{o}bner basis $G'$ including the set $\mathcal{F}$. In addition, any redundant operators in $G'$ are not allowed. A redundant operator is defined as follows.
\begin{defn}\label{def_redundant}
A non-zero operator $f\in I$ is said to be \emph{redundant} to a set $G\subset I$ if there exists an operator $g\in G$ and $t\in \mathcal{T}^{(h)}$ such that $f=tg$ \ and \ $\s(f)<_{\rm sign}\Psih(t)\s(g)$.
\end{defn}

From the definition above, if the the operator $f$ is a redundant to $G$, then the module monomial $\Psih(t)\s(g)$ belongs to $\lsyz{\mathcal{F}}$ since $mg-f=0$ and $\Psih(t)\s(g)>_{\rm sign}\s(f)$. An example of a redundant operator can be seen in Example \ref{eg_redundant}.

\vspace{-1ex}
\section{The F5 Tropical Algorithm}

\subsection{The F5 Criterion}
We begin with the definition of normal pairs.
\begin{defn}\label{def_npair}
Let $g_1,g_2\!\in\! I\backslash\{0\}$ with $\LTh(g_1)\!=\!ch^{\gamma}\x^{\alpha}\xi^{\beta}$ and $\LTh(g_2)\!=\!dh^{\gamma'}\x^{\alpha'}\xi^{\beta'}$.
Let
\vspace{0.2ex}
$$\eta =(\max(\alpha_1,\alpha_1'),\!\cdots\!,\max(\alpha_n,\alpha_n')),\quad
\delta =(\max(\beta_1,\beta_1'),\!\cdots\!,\max(\beta_n,\beta_n')),$$
\vspace{-2ex}
$$\mu =\max(\gamma,\gamma'),\quad u_1=dh^{\mu-\gamma}\x^{\eta-\alpha}\d^{\delta-\beta},\quad \text{and}\quad u_2=ch^{\mu-\gamma'}\x^{\eta-\alpha'}\d^{\delta-\beta'}.$$
\vspace{-2ex}

\noindent
We say that $(g_1,g_2)$ is a \emph{normal pair} if:
\begin{enumerate}[(i)]
\item the $g_i$'s are primitive $\s$-irreducible;
\item $\s(u_ig_i)=\LMh(u_i)\s(g_i)$, $i=1,2$;
\item $\s(u_1g_1)\neq\s(u_2 g_2)$;
\item ${\rm spair}(g_1,g_2):=u_1g_1-u_2g_2$ is not a redundant operator to $G$.
\end{enumerate}
\end{defn}
\noindent
\begin{rmk}\label{rmk_sig_normalpairs}
We remark that if $(g_1,g_2)$ is a normal pair, then we have
$$\s({\rm spair}(g_1,g_2))=\max(\s(u_1g_1),\s(u_2g_2)).$$
Moreover, if $\s(u_1g_1)>_{\rm sign}\s(u_2g_2)$, then $u_1\neq 1$. 
\end{rmk}

The following theorem is an F5 criterion. Unlike in \cite{Arri2011} and in \cite{Vaccon2021}, we do not write explicitly the property: for all $i\in\{1,\cdots,\mn\}$, there exists $g\in G$ such that $\s(g)=\ve_i$. We replace it by the hypotesis of the theorem said that $\mathcal{F}\subseteq G$ and by what we have in Remark \ref{rmk_sigf_i}. Therefore, in principle, this theorem is equivalent to what we have in \cite{Vaccon2021}.

\begin{thm}[F5 criterion]\label{thm_F5criterion}
Let $G=\{g\in I\mid g \text{ is }\s\text{-top-irreducible}\}$ be such that contains the set $\mathcal{F}$.
If for every normal pair $(g_1,g_2)\in G\times G$,  there exist $g\in G$ and $m\in\mathcal{M}^{(h)}$ such that $mg$ is $\s$-top-irreducible and $\Psih(m)\s(g)=\s(mg)=\s({\rm spair}(g_1,g_2))$, then $G$ is an $\s$-Gr\"{o}bner basis of $I$.
\end{thm}
\noindent
We can prove that $G$ satisfying the hypothesis of Theorem \ref{thm_F5criterion} is an $\s$-Gr\"{o}bner basis of $I$ in an analogous way to the proof of Theorem 4.2 in \cite{Vaccon2021}.

\begin{rmk}(Rewritable Criterion).\label{rmk_rewriteable}
If two or more normal pairs are such that their S-pairs have the same signature, then we can freely consider just one of them. This is based on the sufficient condition of the F5 criterion above --- we only care about the signature $\s({\rm spair}(g_1,g_2))$, without involving the corresponding S-pair explicitly. Following \cite{Arri2011}, in practice (see Line \ref{line_sGB_prunerewriteable} of Algorithm \ref{alg_sGB}), we remove $(f,\sigma)$ from $P$ if there exists another operator $f'$ such that $\s(f')=\sigma$ and $\LT_h(f')<_h\LT(f)$.
\end{rmk}

\begin{eg}\label{eg_redundant}
Let $R=D_2^{(h)}(\Q)$ with $3$-adic valuation. Define the tropical term order $\leq_h$ by using $w=(1,1,2,2)$, $\omega=(-1,-1,1,1)$, and the lexicographic order $\prec$ with $\dx\succ\dy\succ x\succ y$. Consider the subset $\mathcal{F}=\{\mathfrak{f}_1=\underline{2\dy}+y, \ \mathfrak{f}_2=\underline{4xy}+3x^2\}$ where $\s(\mathfrak{f}_1)=\ve_1$ and $\s(\mathfrak{f}_2)=\ve_2$, and set $G:=\mathcal{F}$.
\begin{itemize}
\item We can check that $(\mathfrak{f}_2,\mathfrak{f}_1)$ is a normal pair. Moreover, we have ${\rm spair}(\mathfrak{f}_2,\mathfrak{f}_1)=\underline{6x^2\dy}-4xy^2+8xh^2$ with $\mathcal{S}({\rm spair}(\mathfrak{f}_2,\mathfrak{f}_1))=\dy\ve_2$. We can $\s$-top-reduce ${\rm spair}(\mathfrak{f}_2,\mathfrak{f}_1)$ by $\mathfrak{f}_1$ with respect to $\dy\ve_2$, and the result is $\underline{-4xy^2}-3x^2y+8xh^2$. We continue to $\s$-top-reduce $\underline{-4xy^2}-3x^2y+8xh^2$ by $\mathfrak{f}_2$ with respect to $\dy\ve_2$ to obtain $g_2:=8xh^2$. Here, we have the $\s$-top-irreducible operator $g_1$ with respect to $\s(g_1)=\dy\ve_2$.
\item Next, we compute ${\rm spair}(g_1,\mathfrak{f}_1)=-8xyh^2$ with $\mathcal{S}({\rm spair}(g_1,\mathfrak{f}_1))=\dy^2\ve_2$ and also compute ${\rm spair}(g_1,\mathfrak{f}_2)=-24x^2h^2$ with $\mathcal{S}({\rm spair}(g_1,\mathfrak{f}_2))=\xi_y\dy\ve_2$.
\item Add the operator $g_1$ to the set $G$. At this point, we have $G=\{\mathfrak{f}_1,\mathfrak{f}_2,g_1\}$.
\item Since ${\rm spair}(g_1,\mathfrak{f}_2)=-24x^2h^2=-3x\cdot g_1$ and
$$\s({\rm spair}(g_1,\mathfrak{f}_2))\leq_{\rm sign}\mathcal{S}({\rm spair}(g_1,\mathfrak{f}_2))=\xi_y\dy\ve_2<_{\rm sign}-3\xi_x\dy\ve_2=-3\xi_x\s(g_1),$$
${\rm spair}(g_1,\mathfrak{f}_2)$ is redundant to the set $G$ (see Definition \ref{def_redundant}), and hence, the pair $(g_1,\mathfrak{f}_2)$ in not a normal pair.\\
The operator ${\rm spair}(g_1,\mathfrak{f}_1)=-8xyh^2$ $\s$-top-reduces to $0$ by $g_1$ with respect to $\mathcal{S}({\rm spair}(g_1,\mathfrak{f}_1))=\dy^2\ve_2$. Here, we get $\dy^2\ve_2\in\lsyz{\mathcal{F}}$.
\end{itemize}
By the computation above, we can see that the set $G$ satisfies the sufficient condition in Theorem \ref{thm_F5criterion}. Thus, the set $G$ is an $\s$-Gr\"{o}bner besis of the left ideal $R\mathcal{F}$.
\end{eg}

\subsection{Algorithms}
\vspace{-1.5ex}
Algorithm \ref{alg_mainF5} computes Gr\"{o}bner bases $G_1, \cdots,G_{\sn}$ for the left ideal $I_1,\cdots,I_{\sn}$, respectively. Looking at Line \ref{line_mainF5_divalg}, the element $\mathfrak{f}_j$ of the generating set $\mathcal{F}_j$ is reduced by the operator parts of the set $G_i$, an $\s$-Gr\"{o}bner basis of the left ideal $I_i=I_{j-1}$. This reduction uses the tropical division algorithm in the paper of \cite{ADH2023}. If the corresponding reduction result is $0$, then we can skip the computation of $\s$-Gr\"{o}bner basis using Algorithm \ref{alg_sGB} because $\mathfrak{f}_j\in I_{j-1}$, which means that $G_i$ is also an $\s$-Gr\"{o}bner basis of $I_j$. Because of this skipped step, one can see that the index $i\leq j$.

At the beginning of Algorithm \ref{alg_mainF5}, we do not have any information about syzygies, and even we do not have principal syzygies $\mathfrak{f}_i\ve_j-\mathfrak{f}_j\ve_i$ since our ring $R$ is not commutative. What we can do is: collect leading monomials of syzygies during computation in the framework of algorithms (Algorithm \ref{alg_mainF5} and Algorithm \ref{alg_sGB}), starting from the smallest one to the bigger ones.
That is the reason why we always take a normal pair in $P$ with minimal (guessed) signature $a\sigma$ (Line \ref{line_sGB_minsig} of Algorithm \ref{alg_sGB}) as a candidate of either a leading monomial of syzygy or a valid signature; and why at Line \ref{line_mainF5_sigGr} of Algorithm \ref{alg_mainF5}, we always keep the set of known leading monomials of syzygies $S$ for each iteration $j=2,\cdots,\mn$. The notation $(f,a\sigma)$ means an operator $f$ with (a guessed) signature $a\sigma$.

In the framework of algorithms, signatures or guessed signatures are saved with coefficients. These coefficients are obtained from the corresponding module elements during the computation process. There is no deep meaning of these coefficients in Algorithm \ref{alg_mainF5} and Algorithm \ref{alg_sGB} bacause ordering module terms with respect to $\leq_{\rm sign}$ does not depend on the coefficients. The purpose of writing these coefficients is only for explaining Algorithm \ref{alg_sigtopred} regarding to Remark \ref{rmk_preservesigornot}.

Now, let us see Algorithm \ref{alg_sGB}. Line \ref{line_sGB_npairs} computes the set $P$ of all normal pairs $(r,g)$'s, where $g\in G_{i-1}$. To check the condition (ii), $\s(u_ig_i)=\LMh(u_i)\s(g_i)$, of the definition of normal pair (Definition \ref{def_npair}), we apply Proposition \ref{prop_criteria_tdotsigf}. The module monomial $\LMh(u_i)\s(g_i)$ is called a guessed signature of $u_ig_i$. If $\LMh(u_i)\s(g_i)$ is divisible by a known leading monomial of a syzygy in $S$, then the $\LMh(u_i)\s(g_i)$ is a leading monomial of a syzygy, and hence, it does not satisfy the condition (ii) of normal pair. The $\s$-top-reduction at Line \ref{line_sGB_stopred} will also help us to see whether $\sigma:=\LMh(u_i)\s(g_i)$ belongs to $\lsyz{\mathcal{F}}$ or not. If the corresponding reduction result is $0$, then $\sigma$ lies in $\lsyz{\mathcal{F}}$. Otherwise, $\sigma$ is certified as a signature. At Line \ref{line_sGB_prunebysyz}, we remove all elements $(f,a\sigma)\in P$ where $\sigma$ is divisible by an element of the updated $S$. We also remove elements of $P$ (see Line \ref{line_sGB_prunerewriteable}) which satisfy the rewritable criterion based on Remark \ref{rmk_rewriteable}.

The correctness of this algorithm framework (Algorithm \ref{alg_mainF5} and Algorithm \ref{alg_sGB}) is based on Theorem \ref{thm_F5criterion}. The termination of Algorithm \ref{alg_sGB} is guaranted by Proposition \ref{propcontainfiniteGB} since we only save primitive $\s$-irreducible operators which are not redundant elements.

\vspace{-2ex}
\begin{minipage}[t]{0.45\textwidth}
{\small
\begin{algorithm}[H]
\caption{Main framework of tropical F5 algorithm}\label{alg_mainF5}
\KwInput{$\mathcal{F}=\{\mathfrak{f}_1,\cdots,\mathfrak{f}_{\sn}\}\subset R$ {homogeneous},\\
~~~ ordered increasingly by $\leq_h$ on their $\LT_h$}

\KwOutput{$G$: a Gr\"{o}bner basis of the left ideal $I$\\
~~~\qquad\qquad   generated by $\mathcal{F}$}
\BlankLine
$S\leftarrow \emptyset$\\
$G_1\leftarrow\{(\mathfrak{f}_1,\ve_1)\}$\\
$i\leftarrow 1$\\
\SetAlgoVlined
\SetKwBlock{Begin}{for $j=2,\cdots,\mn$ do}{}
\Begin{
  $r\leftarrow$\texttt{Div}$(\mathfrak{f}_j,{\rm op\_part}(G_i))$\label{line_mainF5_divalg}\\
  \SetKwBlock{Begin}{if $r\neq 0$ then}{}
  \Begin{
    $i\leftarrow i+1$\\
    $G_i, S\leftarrow {\rm SigGr}((r,\ve_i), ~G_{i-1}, ~S, ~j)$\label{line_mainF5_sigGr}
  }
}
{\bf Return:} $G:={\rm op\_part}(G_i)$
\end{algorithm}
}
\end{minipage}
~~~
\begin{minipage}[t]{0.51\textwidth}
{\small
\begin{algorithm}[H]
\caption{SigGr(~)}\label{alg_sGB}
\KwInput{$(r,\ve_i)$,\quad $G_{i-1}$,\quad $S$,\quad $j$}

\KwOutput{$G_i$: an $\s$-Gr\"{o}bner basis of the left ideal $I_j$ generated by $\{\mathfrak{f}_1,\cdots,\mathfrak{f}_j\}$\\
\qquad\quad~ \ $S$~: leading terms of syzygies
}
\medskip
\noindent
$[P,S]\leftarrow {\rm NPairs}(\emptyset,\ (r,\ve_i),\ G_{i-1},\ S)$\label{line_sGB_npairs}\\
$G_i\leftarrow G_{i-1}\cup\{(r,\ve_i)\}$\\
\SetAlgoVlined
\SetKwBlock{Begin}{while $P\neq \emptyset$ then}{}
\Begin{
  {{\bf Prune} $P$ by $S$}\label{line_sGB_prunebysyz}\\
  {{\bf Prune} $P$ using rewritable criterion}\label{line_sGB_prunerewriteable}\\
  {\bf Take} $(f,a\sigma)\in P$ with minimal $a\sigma$\label{line_sGB_minsig}\\
  $(f,b\sigma) \leftarrow \s$-top-reduce$(f,a\sigma, G)$\label{line_sGB_stopred}\\
  \SetAlgoLined
  \SetKwBlock{Begin}{if $f\neq 0$ then}{}
  \Begin{
    $[P,S]\leftarrow$ {\small${\rm Update } \ {\rm NPairs}(P,\ (f,b\sigma),\ G_{i},\ S)$}\\
    $G_i\leftarrow G_i\cup\{(f,b\sigma)\}$
  }
  \SetAlgoVlined
  \SetKwBlock{Begin}{else}{}
  \Begin{
    $S\leftarrow S\cup\{b\sigma\}$
  }
}
{\bf Return:} \ $G_i$ , \ $S$
\end{algorithm}
}
\end{minipage}

To $\s$-top-reduce $(f,a\sigma)$ by $G$ in Line \ref{line_sGB_stopred} of Algorithm \ref{alg_sGB}, we use Algorithm \ref{alg_sigtopred}. It is a modification of the tropical division algorithm (privided in the paper of \cite{ADH2023}). In this algorithm, we only allow $\s$-top-reducers in $G$ with respect to $a\sigma$ and reducers from the previous computed operators $q_{j'}$'s, $j'<j$ as a reducer of $q_j$ (see Line \ref{line_sigtopred_reducers}). The termination reason of Algorithm \ref{alg_sigtopred} is the same reason of termination of tropical division algorithm, that is, by the fact that the number of supports $\supp(q_j)$'s for operators $q_j$'s with a fixed degree is finite (see the paper of \cite{ADH2023}). In addition, if there is no a reducer found for $q_j$, then the algorithm is immediately terminated (see Line \ref{line_sigtopred_noreducer}).
The proof that the outputs satisfy (i) and (ii) follows the proof in the tropical division algorithm. Moreover, due to the hypothesis on $G$ --- a Gr\"{o}bner basis up to signature $a\sigma$, the output $r$ is an $\s$-top-irreducible operator with respect to $a\sigma$ (Clause (iii) of the algorithm output).

\begin{algorithm}[H]
\caption{$\s$-Top-Reduction}\label{alg_sigtopred}
\BlankLine
\KwInput{$f\in I$ homogeneous, $a\in K^*$, \ $\sigma\in \mathbb{M}_{\xi}$ such that $\mathcal{S}(f)=a\sigma$;\\
\quad\quad\quad \ $G\!=\!\{g_1,\cdots,g_{s}\}\!\subset\! R$ containing homogeneous elements $g_i$'s of  an $\s$-Gr\"{o}bner basis with $\s(g_i)<a\sigma$.
}

\KwOutput{$H_1,\cdots,H_{s},r\in R$ satisfying:
\begin{enumerate}[(i)]
\item $f=\left(\sum_{i=1}^{s}H_ig_i\right)+r$
\item $\LT_h(H_{i}g_i)\leq_h \LT_h(f)$;
\item there is no $\s$-top-reducer in $G$ for $r$ with respect to $a\sigma$.
\end{enumerate}
}
\BlankLine 
$T\leftarrow G$, \ \ \ $j\leftarrow 0$, \ \ \ $q_0\leftarrow f$, \quad $h_{1,0}\leftarrow0,\cdots,h_{s,0}\leftarrow0$\label{line_sigtopred_init}

\BlankLine
\SetAlgoLined
\SetKwBlock{Begin}{While $q_j\neq 0$ do}{}
\Begin{
    \SetAlgoLined
    \SetKwBlock{Begin}{If \normalfont there is no an $\s$-top-reducer in $G$ for $q_j$ w.r.t. $a\sigma$ AND there is no $g\in T\backslash G$ with $\LM_h(g)|\LM_h(q_j)$ {\bf then}}{}
    \Begin{\label{line_sigtopred_noredcondition}
      {\bf Return:} \ $r:=q_j$,\quad $H_i=h_{i,j}$ ~ for all $i$\label{line_sigtopred_noreducer}
    }
    \SetAlgoVlined
    \SetKwBlock{Begin}{else}{}
    \Begin{
      {\bf Choose} an $\s$-top-reducer $g\in G$ for $q_j$ w.r.t. $a\sigma$ \ OR \ an element $g\in T\backslash G$ with $\LM_h(g)\mid\LM_h(q_j)$; such that $E(q_j,mg)$ is minimum among all such choices, for some $m\in\mathcal{M}^{(h)}$ satisfying $\LM_h(q_j)=\LM_h(mg)$.\label{line_sigtopred_reducers}\\
     \BlankLine
      $c={\rm LC}_h(q_j)/{\rm LC}_h(mg)$\label{line_sigtopred_c}\\
      $q'\leftarrow q_j-c\,mg$\label{line_sigtopred_cancelation}\\
      \BlankLine
      \SetAlgoVlined
      \SetKwBlock{Begin}{If {\normalfont $g=g_k$ for some $1\leq k\leq s$} then}{}
      \Begin{\label{line_sigtopred_caseG}
        $q_{j+1}\leftarrow q'$\label{line_sigtopred_cancelationG};\\
        $h_{k,j+1}\leftarrow h_{k,j}+c\, m$;\\
        $h_{i,j+1}\leftarrow h_{i,j}$ for $i\neq k$
      }
      \SetKwBlock{Begin}{If {\normalfont $g$ was added to $T$ at some previous iteration of the algorithm, namely $g=q_k$ for some $k<j$}, then}{}
      \Begin{\label{line_sigtopred_caseTminG}
        $q_{j+1}\leftarrow \frac{1}{1-c}q'$\label{line_sigtopred_qj1}\\
        $h_{i,j+1}\leftarrow \frac{1}{1-c}(h_{i,j}-c\,h_{i,k})$ for all $i$\label{line_sigtopred_hj1}\\
      }
      \SetAlgoVlined
      \BlankLine
      \SetKwBlock{Begin}{If $E(q_j,m\,g)>0$ then}{}
      \Begin{
        $T\leftarrow T\cup\{q_j\}$
      }
    }
    $j=j+1$
}
{\bf Return}  $r:=q_j$,\quad $H_i=h_{i,j}$ ~  for all $i$.\label{line_sigtopred_rzero}
\end{algorithm}

\vspace{1ex}
\begin{prop}\label{prop_nochangesig}
When Algorithm \ref{alg_sigtopred} outputs the result, we have either
$\sigma\in\lsyz{\mathcal{F}} \ \text{ or } \ \sigma=\s(r)$.
\end{prop}
\begin{proof}
If $f=0$, then $r=q_0=f$ (see Line \ref{line_sigtopred_init} and Line \ref{line_sigtopred_rzero}). Since $a\sigma$ is a guessed signature of $f$, we get $\sigma\in\lsyz{\mathcal{F}}$.

\noindent
Now, assume that $f\neq 0$.
Since $q_0=f$, we have $S(q_0)=a\sigma$.
Let $\q^{(0)}$ be a module element of $R^{\sn}$ such that $\LM_{\rm sign}(\q^{(0)})=a\sigma$ and $\nu(\q^{(0)})=q_0$.

\noindent
When $j=0$ (the first iteration), we have $T\backslash G=\emptyset$, and hence, there are two possibilities about reducers for $q_0$: (1) no reducer from $T$ (the argument/condition in Line \ref{line_sigtopred_noredcondition} is true); (2) there is an $\s$-top-reducer $g'\in G$ for $q_0$ w.r.t. $a\sigma$. If Case (1) occurs, then $r=q_0=f$, $\s(r)=\sigma$ (by the second statement of Theorem \ref{thm_fzeroiff}), and the process terminates. If Case (2) occurs, then we have $g=g'$ and $m=m'$ for some $m'\in \mathcal{M}^{(h)}$ such that $\LM_h(q_0)=\LM_h(m'g')$. Let $\g^{(0)}\in R^{\sn}$ be such that $\LT_{\rm sign}(\g^{(0)})=\s(g)$ and $\nu(\g^{(0)})=g$. By Line \ref{line_sigtopred_cancelation} and Line \ref{line_sigtopred_cancelationG}, we have $q_{1}=q_0-cmg$, and its corresponding module representation is $\q^{(1)}:=\q^{(0)}-cm\g$. Since $g$ is $\s$-top-reducer for $q_0$ with respect to $a\sigma$, we get $\LM_{\rm sign}(\q^{(1)})=a\sigma$. We keep this last equation for the next step ($j=1$).

\noindent
When $j=1$, we have $T\backslash G\neq\emptyset$, and hence, we have three possibilities about reducers for $q_1$: 
\begin{enumerate}[(A)]
\vspace{-0.45ex}
\item no reducer from $T$ (the argument/condition in Line \ref{line_sigtopred_noredcondition} is true);
\vspace{-0.75ex}
\item the reductor is an $\s$-top-reducer $g'\in G$ for $q_1$ w.r.t. $a\sigma$; and
\vspace{-0.75ex}
\item the reductor is an operator $g\in T\backslash G$ ($g=q_0$).
\end{enumerate}
Case (A) is an analogue of Case (1) which result $r=q_1$ and $\s(r)=\sigma$. Case (B) is analogue of Case (2). Write $g=g'$ and $m=m'$ for some $m'\in \mathcal{M}^{(h)}$ such that $\LM_h(q_1)=\LM_h(m'g')$. Let $\g^{(1)}\in R^{\sn}$ be such that $\LT_{\rm sign}(\g^{(1)})=\s(g)$ and $\nu(\g^{(1)})=g$. By the algorithm, we have $q_{2}=q_1-cmg$, and its corresponding module representation is $\q^{(2)}:=\q^{(1)}-cm\g$. Since $g$ is $\s$-top-reducer for $q_1$ with respect to $a\sigma$, we get
$\LM_{\rm sign}(\q^{(2)})=a\sigma.$
For Case (C), write $m=m''$ for some $m''\in\mathcal{M}^{(h)}$ such that $\LM_h(q_1)=\LM_h(m''g)$. Since $\LM_h(q_1)=\LM_h(mg)$ and both $q_1$ and $g$ are homogeneous of the same degree, then we must have $m=1$. By Line \ref{line_sigtopred_cancelation} and Line \ref{line_sigtopred_qj1}, we have
\begin{equation}\label{eq_q2}
q_2=(\frac{1}{1-c})(q_1-cg).
\end{equation}
By the previous step, we can write $\q^{(1)}=a\Phi^{-1}(\sigma) + (\text{lower module terms})$. Since $g=q_0$, we can write $\q^{(0)}=a\Phi^{-1}(\sigma) + (\text{lower module terms})$. Therefore, the corresponding module representation of \eqref{eq_q2} is
\vspace{-0.5ex}
\begin{align*}
\q^{(2)}:&=\frac{1}{1-c}\left(\q^{(1)}-c\q^{(0)}\right)\\
&=\frac{1}{1-c}\left[(a\Phi^{-1}(\sigma) + (\text{lower module terms})) \ \ - \ \ c(a\Phi^{-1}(\sigma) + (\text{lower module terms}))\right]\\
&=\frac{1}{1-c}\left[(1-c)a\Phi^{-1}(\sigma) + (\text{lower module terms})\right]\\
&=a\Phi^{-1}(\sigma) + (\text{lower module terms}).
\end{align*}
We get $\LM_{\rm sign}(\q^{(2)})=a\sigma$. Thus, Case (B) and Case (C) result an operator $q_2$ and its corresponding module representation $\q^{(2)}$ with $\LM_{\rm sign}(\q^{(2)})=a\sigma$.

\noindent
In general, for $j>1$, we have an analogue result. When the process terminates due to Line \ref{line_sigtopred_noreducer}, the result is $r=q_j$ with $\s(r)=\sigma$. When the argument either in Line \ref{line_sigtopred_caseG} or in Line \ref{line_sigtopred_caseTminG} is true, we obtain an operator $q_{j+1}$ and its corresponding module representation $\q^{(j+1)}$ with $\LM_{\rm sign}(\q^{(j+1)})=a\sigma$.

\noindent
The process eventually terminates because $G$ is a finite set of homogeneous operators, and we always take a reducer $g$ with minimal \'{e}cart function value. For detail explanation, we can refer the paper of \cite{ADH2023}. When the iterations end up with $q_l=0$ for some $l$, its corresponding module representation is $\q^{(l)}$ with $\LM_{\rm sign}(\q^{(l)})=a\sigma$. This implies that $\sigma\in\lsyz{\mathcal{F}}$.
\end{proof}

\vspace{-1ex}
\subsection{Implementation Results on Risa/Asir CAS}

We implemented the algorithms presented in the previous subsection in Risa/Asir Computer Algebra System. Some of the data results can be seen in Table \ref{table_implementationresult}. We provide five examples of generating sets $\mathcal{F}$.
\vspace{-1ex}
\noindent
\begin{table}[H]
\caption{Computation Data: Buchberger Algorithm vs  Tropical F5 Algorithm}\label{table_implementationresult}
\centering
\bgroup
\def\arraystretch{1.25}
{\small
\begin{tabular}{|c||c|c||c|c|}
\hline
\multirow{2}{*}{$\mathcal{F}$} & \multicolumn{2}{c||}{Buchberger Algorithm}& \multicolumn{2}{c|}{Tropical F5 Algorithm}\\
\cline{2-5}
 & \parbox{1.3cm}{{\scriptsize\#S-pairs}} & \parbox{2.2cm}{{\scriptsize\#ReductionsTo$0$}} & \parbox{1.85cm}{{\scriptsize\#NormalPairs}} & \parbox{2.5cm}{{\scriptsize\#$\s$-ReductionsTo$0$}}\\
\hline\hline
\multirow{2}{*}{
\parbox{3cm}{
$4x\dy+yz$\\$z\dx+4y\dz$}} & \multirow{2}{*}{$545$} & \multirow{2}{*}{$512$} & \multirow{2}{*}{\textcolor{blue}{$100$}} & \multirow{2}{*}{\textcolor{blue}{$41$}}\\
 &  &  &  & \\
\hline\hline
 \multirow{3}{*}{
\parbox{3cm}{
$2z\dy+3y^2$\\$3y\dx+2x\dz$\\$x\dy\dz+4y^2\dx$}} & \multirow{3}{*}{$541$} & \multirow{3}{*}{$509$} & \multirow{3}{*}{\textcolor{blue}{$126$}} & \multirow{3}{*}{\textcolor{blue}{$42$}}\\
 &  &  &  & \\
 &  &  &  & \\
\hline\hline
 \multirow{3}{*}{
\parbox{3cm}{
$x\dx+6y\dy$\\$2y\dy^2+3z\dx^2$\\$z\dx^2+4x\dz^2+6y\dz^2$}} & \multirow{3}{*}{$623$} & \multirow{3}{*}{$589$} & \multirow{3}{*}{\textcolor{blue}{$243$}} & \multirow{3}{*}{\textcolor{blue}{$92$}}\\
 &  &  &  & \\
  &  &  &  & \\
 \hline\hline
 \multirow{2}{*}{
\parbox{3cm}{
$y\dx\dy+6x^2z$\\$x^2\dy^2+3y^2yz$}} & \multirow{2}{*}{10,195} & \multirow{2}{*}{10,051} & \multirow{2}{*}{\textcolor{blue}{$117$}} & \multirow{2}{*}{\textcolor{blue}{$18$}}\\
 &  &  &  & \\
 \hline\hline
 \multirow{2}{*}{
\parbox{3cm}{
$3y\dy^2+2z^2\dz$\\$2\dx^3\dz+3x^2z\dy$}} & \multirow{2}{*}{5,044} & \multirow{2}{*}{4,944} & \multirow{2}{*}{\textcolor{blue}{$138$}} & \multirow{2}{*}{\textcolor{blue}{$21$}}\\
 &  &  &  & \\
\hline
\end{tabular}
}
\egroup
\end{table}
\vspace{-2ex}

\noindent
{\bf Settings:}\quad $R=D_3^{(h)}(\Q)$ with $3$-adic valuation; the tropical term order $\leq_h$ defined by using $w=(1,1,1,2,2,2)$,\\
\phantom{{\bf Settings:}\quad}$\omega=(-1,-1,-1,1,1,1)$, and the graded lexicographic order $\prec$ with $\d_x\succ\d_y\succ\d_z\succ x\succ y\succ z$.
\vspace{2ex}

From the table above, we can see that all ratios \#ReductionsTo$0$/\#S-pairs by the Buchberger algorithm are greater than $1/2$, and even it is close to $1$. It means that there are too many useless reductions to zero in the computation. In contrast, by the tropical F5 algorithm, the ratios \#$\s$-ReductionsTo$0$/\#NormalPairs are less than $1/2$, and the number of normal pairs considered during the computation is also small enough. Actually, there are many pairs considered during the computation, but many of them are not normal pairs and are eliminated by the syzygy criterion (see Line \ref{line_sGB_prunebysyz} of Algorithm \ref{alg_sGB}) and the rewriteable criterion (see Line \ref{line_sGB_prunerewriteable} of Algorithm \ref{alg_sGB}) in advance before a $\s$-reduction occurs.

As a conclusion, we can say that we have successfully developed an F5 algorithm for computing tropical Gr\"{o}bner bases on the Weyl algebras. We perform only a few reductions to zero, comparing the number of useless reductions to zero in the Buchberger Algorithm. However, we remark that recognizing a normal pair in the set of considered pairs needs an effort during the computation. Pruning the set of pairs $P$ by using rewritable criterion also takes time because we do it for each iteration. 

To use the total order $\leq_{\rm sign}$ on the set $\mathbb{M}_{\xi}$, we require information about leading monomials of syzygies. We can get this information during the computation of $\s$-Gr\"{o}bner basis elements increasingly on signatures. In the other words, at the time we process an S-pair with signature $\sigma$, we already had an $\s$-Gr\"{o}bner basis up to signature $\sigma$ as well as leading monomials of syzygies that are less than $\sigma$. Our assumption on the homogeneous operators $\mathfrak{f}_1,\cdots,\mathfrak{f}_{\sn}$ ordered increasingly on $\LT_h(\mathfrak{f}_i)$'s and giving priority to the ordering indexes of $\ve_i$'s and the total degrees (see Clause (1) and (2) of Definition \ref{def_totalordermodule}) on the definition of $\leq_{\rm sign}$ make compatible with our desire to handle the signature of S-pairs increasingly. Remark \ref{rmk_sig_normalpairs} also supports this situation in the sense that once we compute a new normal pair, the degree of its signature is always greater than the signature of the generating operators.

\vspace{-1ex}
\section*{Acknowledgment}
\vspace{-1ex}
\noindent
This research is supported by MEXT Scholarship in Japan and partially by Innovative Asia Program of JICA. In addition, the second author is supported by KAKENHI 21K03291.

\vspace{-0.5ex}
\bibliographystyle{apa}
\bibliography{main}  






\end{document}